\documentclass[11pt]{article}
\usepackage{amsmath,amsthm,latexsym,amssymb,amsfonts,epsfig, psfrag}

\makeatletter \@addtoreset{equation}{section} \makeatother

\def\be{\begin{equation}}
\def\ee{\end{equation}}
\def\ba{\begin{eqnarray}}
\def\ea{\end{eqnarray}}

\def\Nl{{\mathchoice
{\setbox0=\hbox{$\displaystyle\rm N$}\hbox{\hbox to0pt
{\kern0.4\wd0\vrule height0.9\ht0\hss}\box0}}
{\setbox0=\hbox{$\textstyle\rm N$}\hbox{\hbox to0pt
{\kern0.4\wd0\vrule height0.9\ht0\hss}\box0}}
{\setbox0=\hbox{$\scriptstyle\rm N$}\hbox{\hbox to0pt
{\kern0.4\wd0\vrule height0.9\ht0\hss}\box0}}
{\setbox0=\hbox{$\scriptscriptstyle\rm N$}\hbox{\hbox to0pt
{\kern0.4\wd0\vrule height0.9\ht0\hss}\box0}}}}
\def\Zl{{\mathchoice
{\setbox0=\hbox{$\displaystyle\rm Z$}\hbox{\hbox to0pt
{\kern0.4\wd0\vrule height0.9\ht0\hss}\box0}}
{\setbox0=\hbox{$\textstyle\rm Z$}\hbox{\hbox to0pt
{\kern0.4\wd0\vrule height0.9\ht0\hss}\box0}}
{\setbox0=\hbox{$\scriptstyle\rm Z$}\hbox{\hbox to0pt
{\kern0.4\wd0\vrule height0.9\ht0\hss}\box0}}
{\setbox0=\hbox{$\scriptscriptstyle\rm Z$}\hbox{\hbox to0pt
{\kern0.4\wd0\vrule height0.9\ht0\hss}\box0}}}}
\def\Ql{{\mathchoice
{\setbox0=\hbox{$\displaystyle\rm Q$}\hbox{\hbox to0pt
{\kern0.4\wd0\vrule height0.9\ht0\hss}\box0}}
{\setbox0=\hbox{$\textstyle\rm Q$}\hbox{\hbox to0pt
{\kern0.4\wd0\vrule height0.9\ht0\hss}\box0}}
{\setbox0=\hbox{$\scriptstyle\rm Q$}\hbox{\hbox to0pt
{\kern0.4\wd0\vrule height0.9\ht0\hss}\box0}}
{\setbox0=\hbox{$\scriptscriptstyle\rm Q$}\hbox{\hbox to0pt
{\kern0.4\wd0\vrule height0.9\ht0\hss}\box0}}}}
\def\Rl{{\mathchoice
{\setbox0=\hbox{$\displaystyle\rm R$}\hbox{\hbox to0pt
{\kern0.4\wd0\vrule height0.9\ht0\hss}\box0}}
{\setbox0=\hbox{$\textstyle\rm R$}\hbox{\hbox to0pt
{\kern0.4\wd0\vrule height0.9\ht0\hss}\box0}}
{\setbox0=\hbox{$\scriptstyle\rm R$}\hbox{\hbox to0pt
{\kern0.4\wd0\vrule height0.9\ht0\hss}\box0}}
{\setbox0=\hbox{$\scriptscriptstyle\rm R$}\hbox{\hbox to0pt
{\kern0.4\wd0\vrule height0.9\ht0\hss}\box0}}}}
\def\Cl{{\mathchoice
{\setbox0=\hbox{$\displaystyle\rm C$}\hbox{\hbox to0pt
{\kern0.4\wd0\vrule height0.9\ht0\hss}\box0}}
{\setbox0=\hbox{$\textstyle\rm C$}\hbox{\hbox to0pt
{\kern0.4\wd0\vrule height0.9\ht0\hss}\box0}}
{\setbox0=\hbox{$\scriptstyle\rm C$}\hbox{\hbox to0pt
{\kern0.4\wd0\vrule height0.9\ht0\hss}\box0}}
{\setbox0=\hbox{$\scriptscriptstyle\rm C$}\hbox{\hbox to0pt
{\kern0.4\wd0\vrule height0.9\ht0\hss}\box0}}}}
\def\Hl{{\mathchoice
{\setbox0=\hbox{$\displaystyle\rm H$}\hbox{\hbox to0pt
{\kern0.4\wd0\vrule height0.9\ht0\hss}\box0}}
{\setbox0=\hbox{$\textstyle\rm H$}\hbox{\hbox to0pt
{\kern0.4\wd0\vrule height0.9\ht0\hss}\box0}}
{\setbox0=\hbox{$\scriptstyle\rm H$}\hbox{\hbox to0pt
{\kern0.4\wd0\vrule height0.9\ht0\hss}\box0}}
{\setbox0=\hbox{$\scriptscriptstyle\rm H$}\hbox{\hbox to0pt
{\kern0.4\wd0\vrule height0.9\ht0\hss}\box0}}}}
\def\Ol{{\mathchoice
{\setbox0=\hbox{$\displaystyle\rm O$}\hbox{\hbox to0pt
{\kern0.4\wd0\vrule height0.9\ht0\hss}\box0}}
{\setbox0=\hbox{$\textstyle\rm O$}\hbox{\hbox to0pt
{\kern0.4\wd0\vrule height0.9\ht0\hss}\box0}}
{\setbox0=\hbox{$\scriptstyle\rm O$}\hbox{\hbox to0pt
{\kern0.4\wd0\vrule height0.9\ht0\hss}\box0}}
{\setbox0=\hbox{$\scriptscriptstyle\rm O$}\hbox{\hbox to0pt
{\kern0.4\wd0\vrule height0.9\ht0\hss}\box0}}}}
\newcommand{\fl}{\mathfrak{l}}

\newcommand{\eqa}{\begin{eqnarray}}
\newcommand{\neqa}{\end{eqnarray}}

\usepackage{bbm}

\newtheorem{thm}{Theorem}
\newtheorem{lem}[thm]{Lemma}

\newtheorem{cor}[thm]{Corollary}

\newtheorem{rem}[thm]{Remark}
\newtheorem{defn}[thm]{Definition}

\newcommand{\g}{\ensuremath{\mathbf{g}}}
\newcommand{\GRG}{{\itshape Gen.\ Rel.\ Grav. }}
\newcommand{\CQG}{{\itshape Class.\ Quant.\ Grav. }}
\newcommand{\JMP}{{\itshape J.\ Math.\ Phys. }}
\newcommand{\PR}{{\itshape Phys.\ Rev. }}
\newcommand{\PRL}{{\itshape Phys.\ Rev.\ Lett.\ }}

\title{Encoding cosmological futures with conformal structures}
\author{Philipp A H\"ohn\footnote{Present address: Institute for Theoretical Physics, Utrecht University, Leuvenlaan 4, NL-3584 CE Utrecht, The Netherlands} \ and Susan M Scott\\
\\
\small Centre for Gravitational Physics, Department of Physics, College of Science, \\
\small The Australian National University, Canberra ACT 0200, AUSTRALIA\\
\small E-mail: P.A.Hohn@uu.nl, Susan.Scott@anu.edu.au}

\date{}

\begin{document}

\maketitle

\begin{abstract}
Quiescent cosmology and the Weyl curvature hypothesis possess a
mathematical framework, namely the definition of an \emph{isotropic
singularity}, but only for the initial state of the universe. A
complementary framework is necessary to also encode appropriate
cosmological futures. In order to devise a new framework we
analyse the relation between regular conformal structures and
(an)isotropy, the behaviour and role of a monotonic conformal factor which
is a function of cosmic time, as well as four example cosmologies
for further guidance. Finally, we present our new definitions of
an \emph{anisotropic future endless universe} and an \emph{anisotropic future
singularity} which offer a promising realisation for the new
framework. Their irregular, degenerate conformal structures differ
significantly from those of the \emph{isotropic singularity}. The
combination of the three definitions together could then provide
the first complete formalisation of the quiescent cosmology
concept. For completeness we also present the new definitions of
an \emph{isotropic future singularity} and a \emph{future isotropic universe}. The relation to other approaches, in particular to the somewhat dual dynamical systems approach, and other asymptotic scenarios is briefly discussed.

\end{abstract}
\section{Introduction}

Since the discovery of the large scale isotropy of the cosmic microwave background, there has been a major endeavour to explain this characteristic of our universe, with the aim of eventually understanding its origin and future evolution.

From several competing schools of thought, it was Misner's \emph{chaotic cosmology} concept \cite{Misner} which was the first prominent idea to explain the isotropy of our universe. According to this picture the beginning of the universe was extremely irregular and chaotic. The advantage of this view is the avoidance of having to implement too stringent constraints on the initial state of the universe, in order to obtain its currently observed state through cosmic evolution. In a sense the current state of our universe would be fairly independent of the exact conditions at its beginning. The apparent isotropy of the universe was conjectured to have been produced through subsequent dissipative effects, such as the production and collisions of particles, neutrino viscosity, etc., and is therefore observed because we must live at a sufficiently \emph{late stage} of the evolution of our universe. Some detailed analyses, however, have indicated that \emph{chaotic cosmology} is untenable in its full generality, e.g., the expected entropy production through dissipation seems to be incompatible with the observed photon per baryon number \cite{Collins, Hawking, Barrow1977, Barrow1978, Penrose1979}.

\emph{Quiescent cosmology}, on the other hand, was introduced by Barrow \cite{Barrow1978} as an alternative concept to cope with the difficulties arising in the above scenario, and stands in direct contrast to \emph{chaotic cosmology}. According to this view, the universe must, in fact, have originated in a highly regular and smooth beginning. Gravitational attraction, which becomes dominant on large scales, is responsible for the evolution away from the initial isotropy and homogeneity of the universe. This scenario implies that we observe isotropy as a consequence of living at a somewhat \emph{early stage} of the evolution of our universe. It also suggests that a classical cosmological model describing our own universe should possess an initial singularity which is similar to the one in the isotropic FRW models.

A possible justification for the initial isotropy of the universe is offered by Penrose's notion of \emph{gravitational entropy} \cite{Penrose1979}, i.e., the contribution of the gravitational field to the total entropy, which behaves somewhat anomalously in that it increases with the clumping of matter and becomes maximal in a black hole; this might seem counter-intuitive at first sight when one considers the usual behaviour of entropy with regard to the matter distribution. This behaviour is due to the fact that the gravitational force is (at the classical level) universally attractive. As a consequence of the arrow of time the initial state of the universe must have been one of low entropy, and since matter as well as radiation were presumably initially in thermal equilibrium, it seems reasonable to search for the necessary low entropy constraint in the geometry of space-time itself rather than in a special matter distribution. A low gravitational entropy constraint is provided by a faint degree of clumping and therefore spatial isotropy. This low degree of clumping can be associated with a bounded Weyl curvature \cite{Penrose1979} and therefore provides a connection between Weyl curvature and \emph{gravitational entropy} which is known as the \emph{Weyl Curvature Hypothesis} (WCH). There exist several versions of the WCH, with most authors using the hypothesis that the Weyl curvature vanished initially. Here we will adopt a weaker version, however, namely the expectation that the Weyl curvature would initially be matter (i.e., Ricci) dominated and that the converse would be true at a cosmological future\footnote{This weaker version is necessary for cosmologies with anisotropic future evolution, since there are no known cosmologies which satisfy the stronger version other than the completely isotropic FRW universes.}.

For a detailed theoretical investigation of \emph{quiescent cosmology}, the WCH and their structural implications for cosmological models, we need to utilise a suitable mathematical framework. Goode and Wainwright \cite{GW1985} provided the geometric definition of an \emph{Isotropic Singularity}---henceforth called an \emph{Isotropic Past Singularity} (IPS)---which generalises a considerable amount of previous work on initial singularities and which gives a beautiful formalisation of \emph{quiescent cosmology}, at least for the initial state of our universe. For the motivation of the results concerning cosmological futures presented in this article, and for comparison, it is necessary to summarise the features of an IPS at this point\footnote{This version of the definition is due to S.\ M.\ Scott \cite{GVR} who has removed the inherent technical redundancies of the original definition by Goode and Wainwright.}:

\begin{defn}[\bf{Isotropic past singularity}]\label{ISdef1.def} {A space-time $(\mathcal{M},\mathbf{g})$ is said to admit an isotropic past singularity if there exists a space-time $(\tilde{\mathcal{M}},\tilde{\mathbf{g}})$, a smooth cosmic time function\footnote{A cosmic time function increases along every future-directed causal curve. Hawking and Ellis \cite{HawkEll1973} have proven the important result that a space-time $\left(\mathcal{M},\mathbf{g}\right)$ admits a cosmic time function if and only if it is \emph{stably causal}. Requiring the existence of a cosmic time function $T$ is therefore equivalent to requiring stable causality on $(\tilde{\mathcal{M}},\tilde{\mathbf{g}})$.} $T$ defined on $\tilde{\mathcal{M}}$, and a conformal factor $\Omega\left(T\right)$ which satisfy
\begin{enumerate}
\item $\mathcal{M}$ is the open submanifold $T>0$,
\item $\mathbf{g}=\Omega^{2}\left(T\right)\tilde{\mathbf{g}}$ on $\mathcal{M}$, with $\tilde{\mathbf{g}}$ regular (at least $C^{3}$ and non-degenerate) on an open neighbourhood of $T=0$,
\item $\Omega\left(0\right)=0$ and $\exists \ b>0$ such that $\Omega\in C^{0}[0,b]\cap C^{3}(0,b]$ and $\Omega\left(0,b\right]>0$,
\item $\lambda\equiv \mathop {\lim }\limits_{T \to 0^{+}}L\left(T\right)$ exists, $\lambda\neq 1$, where $L\equiv \frac{\Omega ''}{\Omega}\left(\frac{\Omega}{\Omega '}\right)^{2}$ and a prime denotes differentiation with respect to $T$.
\end{enumerate}}
\end{defn}

We will call the cosmological solution,
$(\mathcal{M},\mathbf{g})$, of the Einstein field equations (EFE),
the \emph{physical space-time}, and, correspondingly, the
conformally related space-time,
$(\tilde{\mathcal{M}},\tilde{\mathbf{g}})$---which, in general,
does not describe a physical universe---will be referred to as the
\emph{unphysical space-time} (see figure\ \ref{IPS.img}).
\begin{figure}[h!]
\centering
\includegraphics[width=0.75\textwidth, height=0.55\textwidth]{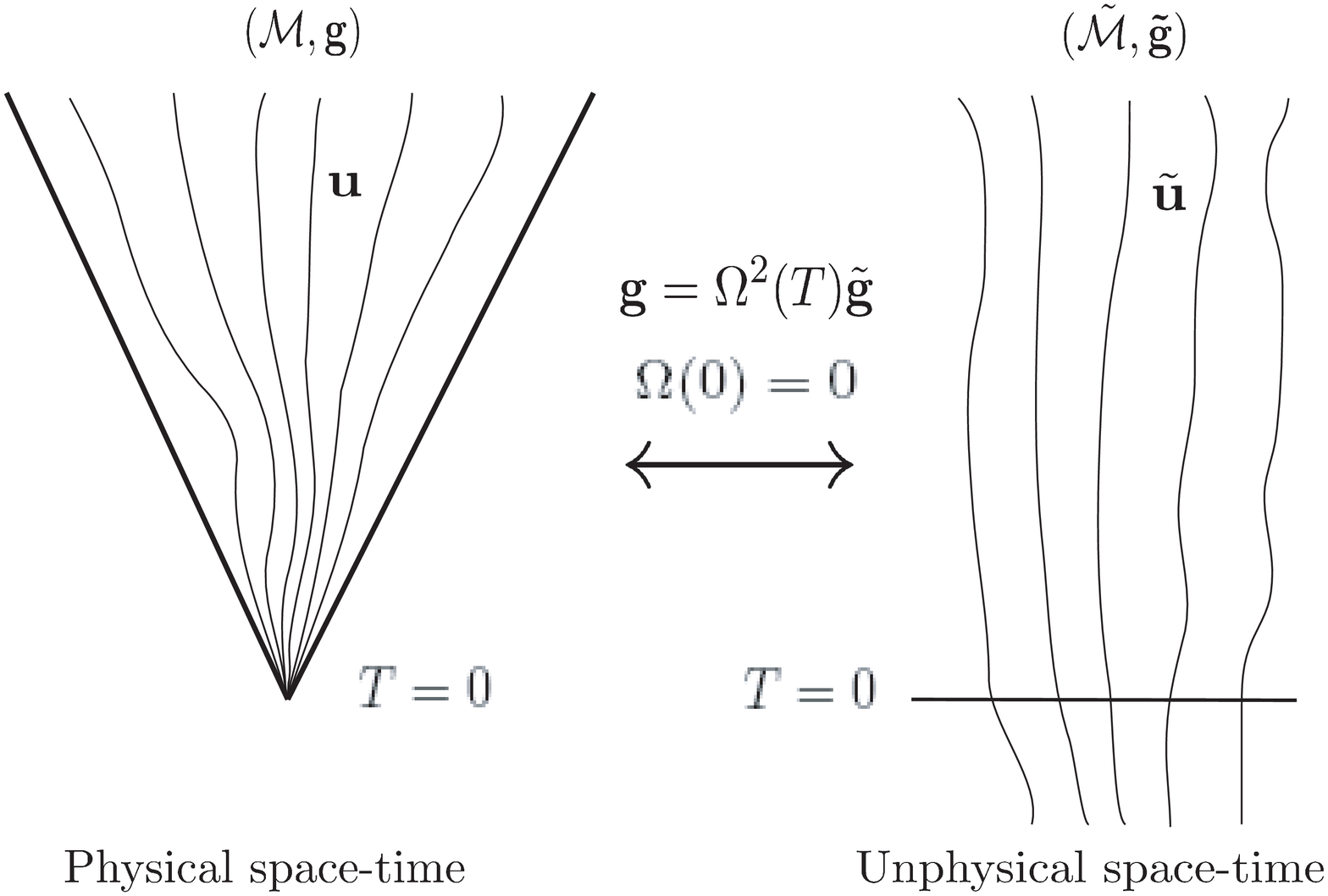}
\caption{{\footnotesize Pictorial interpretation of the definition of an IPS ($\mathbf{u}$ represents the fluid flow of Definition \ref{fluidflow1.def}). }}
\label{IPS.img}
\end{figure}
It is important to note that, due to the form of the definition, one should think here of the \emph{physical space-time} as being produced from the \emph{unphysical} one, rather than vice-versa. In this sense, the initial singularity in the \emph{physical space-time} is due only to the vanishing of $\Omega$ on the regular spacelike hypersurface $T=0$ in $\tilde{\mathcal{M}}$. Consequently, the slice $T=0$ will be referred to as the IPS.

Moving the entire singular behaviour into the vanishing of the conformal factor provides a great analytical advantage, since the regularity in the \emph{unphysical space-time} at $T=0$ guarantees simple derivations of implications of the definition through the conformal relationships for physical quantities \cite{GW1985}. As will be seen, however, this analytical advantage will not fully be realised anymore at an appropriate cosmological future.

In order to explore how isotropic the above definition really is, one can investigate several types of isotropy, all of which can be found in the FRW universes \cite{GW1985,Ellis69,Geoffthesis}. Significant for the subsequent investigations in this article are the following two types:
\begin{description}
\item [(1)] \emph{Weyl isotropy:} $C_{abcd}\equiv 0$, i.e., there do not exist preferred directions due to the principal null directions.
\item [(2)] \emph{Kinematic isotropy relative to the cosmological fluid flow $\mathbf{u}$:} $\sigma =\omega =\dot{u}^{a}\dot{u}_{a}\equiv 0$, i.e., the shear and vorticity eigenvectors, as well as the acceleration vector cannot define preferred directions orthogonal to $\mathbf{u}$.
\end{description}

In fact, it can be shown that Definition \ref{ISdef1.def} alone does not, in general, guarantee such an initial isotropy of a cosmological model \cite{GW1985}, not even asymptotically in the sense that initially the Weyl curvature would be Ricci dominated and that the isotropic kinematic quantity, namely the expansion $\theta$, would dominate over the anisotropic kinematic quantities, such as the vorticity, shear and acceleration. This is trivially the case for the isotropic FRW models which are characterised by a globally vanishing Weyl tensor, vorticity, shear and acceleration and a non-vanishing expansion.

A further constraint on the fluid flow, however, fixes the problem of the initial asymptotic isotropy, namely \cite{GW1985}:

\begin{defn}[\bf{fluid congruence}]\label{fluidflow1.def}  {With any unit timelike congruence $\mathbf{u}$ in $\mathcal{M}$ we can associate a unit timelike congruence $\tilde{{\mathbf{u}}}$ in $\tilde{\mathcal{M}}$ such that
\begin{eqnarray}
\tilde{{\mathbf{u}}}=\Omega\mathbf{u} \qquad \mbox{in} \qquad
\mathcal{M}.
\end{eqnarray}
\begin{description}
\item[(a)] If we can choose $\tilde{{\mathbf{u}}}$ to be regular (at least $C^{3}$) on an open neighbourhood of $T=0$ in $\tilde{\mathcal{M}}$, we say that $\mathbf{u}$ is \emph{regular at the isotropic past singularity.}
\item[(b)] If, in addition, $\tilde{{\mathbf{u}}}$ is orthogonal to $T=0$, we say that $\mathbf{u}$ is \emph{orthogonal to the isotropic past singularity.}
\end{description}}
\end{defn}

Including this definition---which might seem a bit \emph{ad hoc} at first sight, but which is actually fulfilled in a wide class of models---is crucial for the framework of an IPS. Goode and Wainwright \cite{GW1985} proved that for a unit timelike congruence $\mathbf{u}$ which is regular at, and orthogonal to the IPS (and which always exists), an asymptotic isotropy holds, in the sense that
\begin{eqnarray}
\mathop {\lim }\limits_{T \to 0^{+}}K = 0, \qquad \mbox{where} \qquad K=\frac{C_{abcd}C^{abcd}}{R_{ef}R^{ef}}  ,\label{K.eqn} \\
\mathop {\lim }\limits_{T \to
0^{+}}\frac{\sigma^{2}}{\theta^{2}}=0, \qquad \mathop {\lim
}\limits_{T \to 0^{+}}\frac{\omega^{2}}{\theta^{2}}=0, \qquad
\mathop {\lim }\limits_{T \to
0^{+}}\frac{\dot{u}^{a}\dot{u}_{a}}{\theta^{2}}=0.
\label{kiniso.eqn}
\end{eqnarray}
The scalar $K$ is the simplest and most frequently used measure of the relative significance of the Weyl and Ricci curvature and becomes important in probing the WCH. A discussion of $K$ and related ratios in a similar context and of their possible significance as a measure of gravitational entropy may be found in \cite{GoodeColeyWainwright1992,BarrowHervik2002,LUW2006} and references therein.

The above ratios offer the desired initial asymptotic isotropy;
the framework of the IPS therefore provides a formalisation for
\emph{quiescent cosmology} and the WCH, at least for the initial
state of the universe, and it is the combination of Definitions
\ref{ISdef1.def} and \ref{fluidflow1.def} which we mean when we
refer to the IPS. Since the definition does not necessarily imply
an exact initial isotropy, there is room for cosmologies with an
anisotropic future evolution, in which we are subsequently mainly
interested.

There are a number of other significant physical implications of
this definition e.g., a possible implementation of galaxy formation into
the theory, explanations for the \emph{uniqueness} and
\emph{flatness problem} \cite{GoodeColeyWainwright1992} and
information about curvature and the kinematics
\cite{GW1985,GVR,Geoffthesis,Scott2000,Scott2002,Nolan2001}. It was also
argued that \emph{quiescent cosmology} once formulated in the
framework of the IPS may be a viable alternative scenario to
cosmic inflation \cite{GoodeColeyWainwright1992}. Furthermore,
much is already known about the characterising feature for an IPS
to occur in a cosmology \cite{Geoffthesis}. A list of example
cosmologies which admit an IPS may be found in
\cite{Geoffthesis,Scott1998}.

In spite of this success for giving a framework for a possible initial
state of the universe, the definition of an IPS does not guarantee
a future evolution of a respective cosmology which is compatible
with the ideas of the WCH and \emph{quiescent cosmology}.
Projecting these ideas forward in time, it is clear that the local
degree of clumping of matter, and therefore the local anisotropy,
in the universe must increase if the cosmological future is to be
a high-entropy state. Since FRW universes are globally isotropic
and a subclass of them admit an IPS \cite{Geoffthesis}, it becomes
clear that the definition of an IPS is not sufficient, in itself,
to guarantee a future evolution which is consistent with the WCH.

Consequently, to complete the formalisation of \emph{quiescent
cosmology} and the WCH, it is also necessary to encode possible
final states of our universe, thereby providing a possible full
mathematical formulation of these two schools of thought for the
first time. We need a complementary framework which is independent
of models, coordinates and the equation of state of the source of
the gravitational field. In fact, since a recollapsing, as well as
an ever-expanding universe are both compatible with our current
knowledge of the universe, we would like to find, ideally, new
geometric definitions covering both scenarios.

In order to devise a new framework in terms of conformal
structures, we begin, in section 2, by analysing the connection
between such conformal rescalings and the two types of isotropy
discussed above, and showing that a completely regular conformal structure (with vanishing or diverging conformal factor) for an asymptotic state of the universe excludes the possibility that it is anisotropic. As further guidance and motivation for the
construction of the new definitions, we explore, in section
\ref{examples}, four specific example cosmologies for their
conformal relations. Our promising new frameworks for
cosmological futures are presented in section \ref{newdefs}, in
which definitions for both an anisotropic ever-expanding and an
anisotropic recollapsing universe are given along with two
definitions for isotropic cosmological futures. A new feature of the former two cases is the employment of conformal structures which are not completely regular. Essential
technical information about conformal factors, which is necessary
in proving implications of these conformal structures, but does
not represent the key information of the main body of this
article, is given in \ref{conffac}. Our discourse will close in
section \ref{discussion} with a summary of the presented work, a discussion about its relation to other approaches--in particular to the somewhat dual dynamical systems approach--and other asymptotic scenarios, as well as an outlook to a forthcoming paper. For reasons of space the present article will solely display the results forming the foundation of the new framework, while the second paper will demonstrate the derivation of results and implications using the new definitions and provide further example cosmologies. A brief outline of this topic
may furthermore be found in the essay \cite{GRF2007}.

\section{Conformal structures and isotropy}\label{confiso}

In this section we will analyse the relation between conformal
structures, with a conformal factor which is merely a function of
cosmic time, and the two types of isotropy mentioned above. We
will use reasonable constraints on the conformal factor, e.g., we
want to avoid an oscillatory behaviour of it near $T=0$ and
require a monotonic behaviour. Pathological behaviour of the
physical metric shall again be entirely encoded in the conformal
factor; this time, however, we will treat the cases of both a
vanishing and a diverging conformal factor\footnote{For brevity,
we will write $\Omega(0)=+\infty$ for $\mathop {\lim }\limits_{{T}
\to 0^{\pm}}{\Omega}\left({T}\right)=+\infty$ in the remainder of
this article (the limit may be taken along the flow lines of the
cosmological fluid).}. Furthermore, we will cover both initial and
final cosmological scenarios in this section, i.e., the cosmic
time function is allowed to approach $0$ from above (initial
state) or below (final state). Interestingly, the asymptotic
isotropy expressed in the ratios of (\ref{K.eqn}) and
(\ref{kiniso.eqn}) holds in all these cases\footnote{Note that
some of the following results are a generalisation of Theorems 3.1
and 3.3 of \cite{GW1985}.}.

We begin by analysing the relative significance of the Weyl to the
Ricci curvature in terms of the scalar $K$ (see (\ref{K.eqn})).

\begin{thm}[\bf{K-Theorem}]\label{K.thm}{Let $({\mathcal{M}},\mathbf{{g}})$ and $(\tilde{\mathcal{M}},\mathbf{\tilde{g}})$ be two space-times which are related via the conformal rescaling $\mathbf{g}={\Omega}^{2}({T})\mathbf{\tilde{g}}$, where ${T}$ is a smooth cosmic time function defined on $\tilde{\mathcal{M}}$ and $\mathbf{\tilde{g}}$ is non-degenerate and at least $C^{2}$ on an open neighbourhood of ${T}=0$. Let, furthermore, one of the following conditions be true:
\begin{enumerate}
\item ${T}\rightarrow 0^{-}$, ${\Omega}(0)=\infty$ and ${\Omega}$ is positive, $C^{2}$ and strictly monotonically increasing on some interval $[-c,0)$,
$c>0$.
\item ${T}\rightarrow 0^{-}$, $\Omega(0)=0$ and ${\Omega}$ is positive, $C^{2}$ and strictly monotonically decreasing on some interval $[-c,0)$, $c>0$, and continuous on $[-c,0]$.
\item ${T}\rightarrow 0^{+}$, ${\Omega}(0)=\infty$ and ${\Omega}$ is positive, $C^{2}$ and strictly monotonically decreasing on some interval $(0,c]$, $c>0$.
\item ${T}\rightarrow 0^{+}$, ${\Omega}(0)=0$ and ${\Omega}$ is positive, $C^{2}$ and strictly monotonically increasing on some interval $(0,c]$, $c>0$, and continuous on $[0,c]$.
\end{enumerate}
If ${\lambda}$ exists and ${\lambda}\neq1$, then $\mathop {\lim }\limits_{{T} \to 0^{\pm} }K=0$.}
\end{thm}

\begin{proof}
See \ref{proofs}.
\end{proof}

The K-Theorem clearly demonstrates that, if we retain the
reasonable constraints on $\Omega$, we need to exclude regular
conformal metrics from the definition of an initial or final
cosmological state where an asymptotic Weyl anisotropy is supposed
to hold. The proof of the theorem provides us with some further
results, concerning $C_{abcd}C^{abcd}$ and $R_{ab}R^{ab}$, the
proofs of which can be found in \ref{proofs}.

\begin{cor}\label{regmet1.thm}{Suppose condition (ii) or (iv) of Theorem \ref{K.thm} is true. Then, $\mathop {\lim }\limits_{{T} \to 0^{\pm} }R_{ab}R^{ab}=+\infty$, and $\mathop {\lim }\limits_{{T} \to 0^{\pm} }|C_{abcd}C^{abcd}|=+\infty$ (unless possibly if $\mathop {\lim }\limits_{{T} \to 0^{\pm} }\tilde{C}_{abcd}\tilde{C}^{abcd}=0$).}
\end{cor}

\begin{cor}\label{regmet3.thm}{Let condition (i) or (iii) of Theorem \ref{K.thm} be valid. Furthermore, let ${\lambda}\neq1,2,+\infty$, or let ${\lambda}=2$ and ${L}$ be strictly monotonic on some interval $[-c,0]$, $c>0$, if $T\rightarrow 0^{-}$, or on an interval $[0,c]$, $c>0$, if $T\rightarrow 0^{+}$. Then, if  ${M}_{0}\equiv\mathop {\lim }\limits_{{T} \to 0^{\pm} }{M}$, where $M\equiv\frac{{\Omega}'}{{\Omega}^{2}}$,
\begin{eqnarray}
\mathop {\lim }\limits_{{T} \to 0^{\pm} }R_{ab}R^{ab}=\begin{cases} 0 &
\mbox{if ${M}_{0}=0$,}
\\
finite > 0  & \mbox{if $0<|{M}_{0}|<\infty$,}\\
\infty & \mbox{if  $|{M}_{0}|=\infty$.}
\end{cases}
\end{eqnarray}}
\end{cor}

\begin{cor}\label{regmet4.thm}{Let condition (i) or (iii) of Theorem \ref{K.thm} be valid. Furthermore, let ${\lambda}=+\infty$. Then, $\mathop {\lim }\limits_{{T} \to 0^{\pm} }R_{ab}R^{ab}=\infty$. }
\end{cor}

\begin{cor}\label{regmet5.thm}{Let condition (i) or (iii) of Theorem \ref{K.thm} be valid. Then, $\mathop {\lim }\limits_{{T} \to 0^{\pm}}C_{abcd}C^{abcd}=0$.}
\end{cor}

We will now turn to the kinematic isotropy mentioned in the introduction. As one would expect, under the same conditions as above, we also find an asymptotic kinematic isotropy for a unit timelike congruence which is regular at, and orthogonal to the slice $T=0$.

\begin{thm}[\bf{asymptotic kinematic isotropy}]\label{asykiniso.thm}{Let one of the four conditions of Theorem \ref{K.thm} be valid. Then, for a unit timelike congruence $\mathbf{u}$, which is regular at, and orthogonal to the slice $T=0$, we find that
\begin{eqnarray}
\mathop {\lim }\limits_{{T} \to
0^{\pm}}\frac{\sigma^{2}}{\theta^{2}}=0, \qquad \mathop {\lim
}\limits_{{T} \to 0^{\pm}}\frac{\omega^{2}}{\theta^{2}}=0, \qquad
\mathop {\lim }\limits_{{T} \to
0^{\pm}}\frac{\dot{u}^{a}\dot{u}_{a}}{\theta^{2}}=0.
\end{eqnarray} }
\end{thm}

\begin{proof}
See \ref{proofs}.
\end{proof}

Again, the proof of this theorem provides us with some results
concerning the behaviour of the kinematic quantities, the proofs
of which may be found in \ref{proofs}.

\begin{cor}\label{regkin1.thm}{Let condition (ii) or (iv) of Theorem \ref{K.thm} be valid. Then, for a unit timelike congruence $\mathbf{u}$, which is regular at the slice $T=0$, $\mathop {\lim }\limits_{{T} \to 0^{\pm}}\theta^{2} =\infty$ and $\mathop {\lim }\limits_{{T} \to 0^{\pm}}\dot{u}^{a}\dot{u}_{a}=\infty$ (unless possibly if $\mathbf{u}$ is orthogonal to ${T}=0$ and $\mathop {\lim }\limits_{{T} \to 0^{\pm}}\dot{\tilde{u}}^{a}\dot{\tilde{u}}_{a}=0$, where $\ \tilde{} \ $ denotes quantities of $(\tilde{\mathcal{M}},\mathbf{\tilde{g}})$). Also $\mathop {\lim }\limits_{{T} \to 0^{\pm}}\sigma^{2}=\infty$ and $\mathop {\lim }\limits_{{T} \to 0^{\pm}}\omega^{2}=\infty$ (unless possibly if $\mathop {\lim }\limits_{{T} \to 0^{\pm}}\tilde{\sigma}^{2}=0$ or $\mathop {\lim }\limits_{{T} \to 0^{\pm}}\tilde{\omega}^{2}=0$).}
\end{cor}

\begin{cor}\label{regkin3.thm}{Let condition (i) or (iii) of Theorem \ref{K.thm} hold. Then, for a unit timelike congruence $\mathbf{u}$, which is regular at the slice $T=0$,$\mathop {\lim }\limits_{{T} \to 0^{\pm}}\sigma^{2}=\mathop {\lim }\limits_{{T} \to 0^{\pm}}\omega^{2}=0$.}
\end{cor}

\begin{cor}\label{regkin4.thm}{Suppose condition (i) or (iii) of Theorem \ref{K.thm} holds. Then, for a unit timelike congruence $\mathbf{u}$, which is regular at the slice $T=0$, if ${\lambda}\neq 1,2$,
\begin{eqnarray}
\mathop {\lim }\limits_{{T} \to 0^{\pm}}\theta=\begin{cases} 0 &
\mbox{if ${M}_{0}=0$,}
\\
\mp \vartheta,\qquad \mbox{where} \qquad 0<\vartheta<\infty   & \mbox{if $0<|{M}_{0}|<\infty$,}\\
\mp\infty & \mbox{if  $|{M}_{0}|=\infty$,}
\end{cases}
\end{eqnarray}}
and if, furthermore, $\mathbf{u}$ is not orthogonal to
${T}=0$,
\begin{eqnarray}
\mathop {\lim }\limits_{{T} \to
0^{\pm}}\dot{u}^{a}\dot{u}_{a}=\begin{cases} 0 & \mbox{if ${M}_{0}=0$,}
\\
\alpha,\qquad \mbox{where} \qquad 0<\alpha<\infty   & \mbox{if $0<|{M}_{0}|<\infty$,}\\
\infty & \mbox{if  $|{M}_{0}|=\infty$.}
\end{cases}
\end{eqnarray}
\end{cor}

The above results clearly demonstrate that, similarly to the situation at an IPS, we unavoidably obtain an asymptotic Weyl and kinematic istropy at ${T}=0$, if we employ conformal structures with regular conformal metrics and reasonable constraints on the conformal factor (note that the properties of the conformal factor play a key role here), irrespective of whether the cosmic time function approaches $0$ from below or above, or whether ${\Omega}$ diverges or vanishes at ${T}=0$. 

Consequently, we can infer that conformal structures for less symmetric asymptotic states, such as asymptotically anisotropic cosmological futures, are required to be significantly different and to involve irregularities, such as e.g., degeneracy. In particular, this leads us to the important conclusion that: if we are to follow the ideas of \emph{quiescent cosmology} and the WCH with conformal structures involving a conformal factor which is a function of cosmic time, then choosing a conformal structure with regular conformal metric for the initial state of the universe necessarily produces the desired asymptotic isotropy and a conformal structure with irregular conformal metric is required for the anisotropic future infinity or future singularity of the universe. For the latter cases, the analytical advantage of the IPS definition cannot be sustained in the previous manner. In the forthcoming article, however, we will provide evidence that such degenerate structures are, in fact, useful \cite{HoehnScott2}. 

The kinematic and curvature quantities do not necessarily show a
behaviour which is compatible with a cosmological singularity in
the case $\Omega(0)=\infty$, e.g., $\theta$ and $R_{ab}R^{ab}$ do
not, in general, diverge. On the other hand, the case
${\Omega}(0)=0$ necessarily leads to a singular behaviour of the
physical quantities, independently of whether $T$ approaches $0$
from below or above. This state of affairs justifies the
definition of an IPS, which involves the choice of $\Omega(0)=0$
(instead of $\Omega(0)=\infty$) and of a regular conformal metric,
by indicating that the form of a conformal structure with the
desired isotropic and singular physical implications is fairly
constrained.

The presented line of arguments is furthermore important for the general understanding of the problem; it puts the IPS in a broader context and emphasises its special nature from a geometric point of view by implying that different conditions in the conformal structure either possibly lead to less symmetry or to an isotropic asymptotic state which does not necessarily correspond to a singularity. The analysis of some example cosmological models will explicitly demonstrate the conclusions of this section and provide further guidance in elaborating a new framework.

\section{Conformal structures in example cosmologies}\label{examples}

Specific example space-times provide valuable guidance and motivation in the elaboration of reasonable definitions for cosmological futures. We will therefore probe a few example cosmological models for conformal rescalings with a conformal factor as a function of a cosmic time and for the (an)isotropic behaviour of the kinematic and curvature quantities at the respective cosmological futures.

Examining proper time along the fluid flow will become useful in determining whether a metric pathology occurring at the cosmological future actually corresponds to a physical space-time singularity, i.e., whether it occurs at finite or infinite values of proper time along the flow. The proper time $\tau$ for a timelike curve with tangent vector $u^{a}(s)$ and parametrisation $s$ is given by
\begin{eqnarray}
  \tau =\int (-u_{a}u^{a})^{1/2}ds.
\end{eqnarray}
If the fluid flow is parametrised by $\tau$, we find that $u_{a}u^{a}=-1$. All the cosmological models to follow are presented in comoving normal coordinates, hence proper time along the (normalised) fluid flow can be determined via
\begin{eqnarray}
  \tau =\int \frac{d\tau}{dt}dt,\label{proptime.eqn}
\end{eqnarray}
where $t$ is the comoving coordinate time and $\frac{d\tau}{dt}$ can be obtained from the components for $\mathbf{u}$.

Several scalar quantities have been analysed with the aid of the program GRTensorII which runs under Maple. Since some calculations are somewhat extensive, only the results will be presented here. Entities in the forthcoming conformal space-times, as well as the new conformal factor $\bar{\Omega}$ and the cosmic time function $\bar{T}$ itself, will be equipped with a $\bar{} \ $ from now on to avoid confusion with the respective entities at the IPS.

\subsection{The \emph{big crunch} in two closed FRW models}

Although the isotropic FRW cosmologies trivially possess a vanishing $K$ throughout and do not allow for any anisotropies to evolve, we will, nevertheless, analyse two simple cases here for technical interest, namely to reveal a few structural characteristics which distinguish these models from other more physically realistic cosmologies.

For the investigations of \emph{big crunch} singularities we are solely interested in \emph{closed} (k=+1) FRW universes. In this case it will only be possible to find conformal relations in which the conformal factor vanishes at the singularity, due to the vanishing of the scale factor.

Recall that in synchronous coordinates the line-element of the FRW metric is given by
\begin{eqnarray}
ds^{2}&=&-dt^{2}+a^{2}\left(t\right)d\sigma^{2},\label{FRWmetric.eqn}
\end{eqnarray}
where $a(t)$ denotes the scale factor and $d\sigma^{2}$, which is independent of time, is the 3-metric of a spacelike hypersurface of constant curvature.

\subsubsection{A radiation filled, \emph{closed} FRW universe}\label{radFRW}

A radiation filled, \emph{closed} FRW model, with equation of state $p=\frac{1}{3}\mu$, where $\mu$ denotes the energy density, and $p$ denotes the pressure, possesses the scale factor given by
\begin{eqnarray}
a\left(t\right)=C\sqrt{1-\left(1-\frac{t}{C}\right)^{2}},\qquad
\mbox{with}\qquad C^{2}=\frac{\mu a^{4}}{3}=\mbox{const},\qquad
C>0.
\end{eqnarray}
Thus, $a=0$ for $t_{s_{1}}=0$ and $t_{s_{2}}=2C$, where
$t_{s_{2}}$ corresponds to the \emph{big crunch}, which according
to (\ref{proptime.eqn}), occurs at finite proper time. Rewriting
the scale factor as
\begin{eqnarray}
a\left(t\right)=C\sqrt{\left(2-\frac{t}{C}\right)\frac{t}{C}},
\end{eqnarray}
and choosing a cosmic time function $\bar{T}$ such that $\bar{T}=0$ $\Leftrightarrow$ $t=t_{s_{2}}$ and which approaches zero from below:
\begin{eqnarray}
\bar{T}=-\sqrt{2-\frac{t}{C}}, \qquad
t\in[0,2C]\,,\label{radFRW0.eqn}
\end{eqnarray}
yields the following conformal relationship for
(\ref{FRWmetric.eqn})
\begin{eqnarray}
ds^{2}=\bar{\Omega}^{2}\left(\bar{T}\right)\left[-4d\bar{T}^{2}+(2-\bar{T}^{2})d\sigma^{2}\right],
\qquad \mbox{where} \qquad
\bar{\Omega}\left(\bar{T}\right)=-C\bar{T}.\label{radFRW1.eqn}
\end{eqnarray}
The conformal metric $d\bar{s}^{2}$ in the square brackets is $C^{\infty}$ and non-degenerate at $\bar{T}=0$.
The conformal factor $\bar{\Omega}$ vanishes at $\bar{T}=0$ and is positive for $\bar{T}<0$ and $C^{\infty}$ for $\bar{T}\leq0$. The derivative ratios of $\bar{\Omega}$, as given in Definition \ref{ISdef1.def}, behave as
\begin{eqnarray}
\mathop {\lim }\limits_{\bar{T} \to 0^{-}}\frac{\bar{\Omega}
'}{\bar{\Omega}}=\mathop {\lim }\limits_{\bar{T} \to
0^{-}}\frac{1}{\bar{T}}= -\infty, \qquad
\bar{L}=\frac{\bar{\Omega}
''\bar{\Omega}}{{\bar{\Omega'}}^{2}}\equiv 0.
\end{eqnarray}

  (\ref{radFRW1.eqn}) implies that the physical fluid flow diverges as $\bar{T}\rightarrow 0^{-}$, since it is given by
\begin{eqnarray}
\mathbf{u}=-\frac{1}{2C\bar{T}}\frac{\partial}{\partial\bar{T}}.
\end{eqnarray}
The unphysical fluid vector is $C^{\infty}$ everywhere and takes the following form
\begin{eqnarray}
\bar{{\mathbf{u}}}=\bar{\Omega}\mathbf{u}=\frac{1}{2}\frac{\partial}{\partial\bar{T}}.
\end{eqnarray}

As one would expect for a \emph{big crunch}, the expansion scalar of the physical congruence, as well as the Ricci curvature, diverge as $\bar{T}\rightarrow 0^{-}$,
\begin{eqnarray}
\mathop {\lim }\limits_{\bar{T} \to 0^{-}}\theta= -\infty, \qquad
\mathop {\lim }\limits_{\bar{T} \to 0^{-}}R_{ab}R^{ab}= \infty.
\label{radFRW2.eqn}
\end{eqnarray}
Thus the timelike fluid congruence encounters a scalar polynomial curvature singularity in finite proper time. The conformal structure and its properties are completely analogous to that of the IPS, which in this case is, obviously, no surprise since this closed FRW model is symmetric in proper time with respect to the point of maximal expansion. Due to this symmetry we could also have chosen a time-symmetric conformal factor, namely $\bar{\Omega}\left(\bar{T}\right)=C\sqrt{1-\sin^2(\bar{T}+\pi/2)}$, and an appropriate cosmic time function, $\bar{T}=\arcsin(t/C-1)-\pi/2$, which together satisfy exactly the same conditions as the previous choice. We therefore see in this simple example, that there exists a considerable freedom in choosing the explicit cosmic time function and conformal factor, which result in essentially the same conformal structure. This state of affairs holds for all conformal structures considered in this article which therefore do not uniquely define these two functions.

For technical interest, considering the possible definition of a future singularity, it is instructive to analyse the behaviour at $\bar{T}=0$ of some unphysical quantities of the conformally related space-time. Due to the regularity of $d\bar{s}^{2}$ one finds that $\bar{R}_{ab}\bar{R}^{ab}$ approaches a finite value as $\bar{T}\rightarrow 0^{-}$, while $\bar{\theta}$ vanishes asymptotically and $\bar{C}_{abcd}\bar{C}^{abcd}$ vanishes identically (the FRW models being conformally flat).

\subsubsection{A dust, \emph{closed} FRW universe}\label{dustFRW}

The scale factor of a \emph{closed}, dust ($p=0$) FRW universe is given by ($\phi(t)$ is a \emph{development angle})
\begin{eqnarray}
a\left(\phi\right)=\frac{\tilde{C}}{2}\left(1-\cos\phi\right), \qquad \mbox{where} \qquad t=\frac{\tilde{C}}{2}\left(\phi-\sin\phi\right) \label{dustFRW1.eqn}\\
 \mbox{and} \qquad \tilde{C}=\frac{\mu a^{3}}{3}=\mbox{const}, \qquad \tilde{C}>0.
\end{eqnarray}
Hence, $a=0$ $\Leftrightarrow$ $\phi=2z\pi$, with $z\in\mathbb{Z}$. Furthermore,
\begin{eqnarray}
dt&=&\frac{\tilde{C}}{2}\left(1-\cos\phi\right)d\phi=a\left(\phi\right)d\phi\label{dustFRW2.eqn}.
\end{eqnarray}
Equation (\ref{dustFRW2.eqn}) indicates that $\phi$ is a strictly monotonically increasing function of $t$ on an interval $[\frac{\tilde{C}}{2}2z\pi,\frac{\tilde{C}}{2}2(z+1)\pi]$, with $z\in\mathbb{Z}$, thus we can use $\phi$ as a cosmic time function. We choose $\phi\in [-2\pi ,0]$, in order to have the zero-point of the cosmic time function at the future singularity\footnote{One could also choose $t=0$ as the past singularity and reset $\bar{T}=\phi -2\pi$, but the conformal structure remains the same.}, and set
\begin{eqnarray}
\bar{T}=\phi.
\end{eqnarray}
The \emph{big crunch} occurs at $\bar{T}=0$. Using this cosmic
time function and (\ref{FRWmetric.eqn}) leads to the following
conformal structure with time-symmetric conformal factor
\begin{eqnarray}
ds^{2}=\bar{\Omega}^{2}\left(\bar{T}\right)\left[-d\bar{T}^{2}+d\sigma^{2}\right],\label{dustFRW3.eqn}
\\ \mbox{where} \qquad
\bar{\Omega}\left(\bar{T}\right)=\frac{\tilde{C}}{2}\left(1-\cos\bar{T}\right)=a(\bar{T})\,.
\end{eqnarray}
The conformal metric $d\bar{s}^{2}$ in the square brackets is clearly $C^{\infty}$ and non-degenerate on $[-2\pi,0]$. The conformal factor $\bar{\Omega}$ vanishes for both $\bar{T}=-2\pi$ and $\bar{T}=0$, since it is the scale factor, which is positive on $(-2\pi,0)$ and $C^{\infty}$ on $[-2\pi,0]$.

The conformal factor is found to behave as
\begin{eqnarray}
\mathop {\lim }\limits_{\bar{T} \to 0^{-}}\frac{\bar{\Omega}
'}{\bar{\Omega}}&=&-\infty, \qquad \mathop {\lim }\limits_{\bar{T}
\to 0^{-}}\frac{\bar{\Omega}
''\bar{\Omega}}{{\bar{\Omega'}}^{2}}=\frac{1}{2}.
\end{eqnarray}
  (\ref{dustFRW3.eqn}) yields the following expressions for the physical and unphysical fluid flows
\begin{eqnarray}
\mathbf{u}=\frac{2}{\tilde{C}(1-\cos\bar{T})}\frac{\partial}{\partial\bar{T}},
\qquad \bar{{\mathbf{u}}}=\frac{\partial}{\partial\bar{T}},
\end{eqnarray}
which shows, once again, that the physical flow diverges as $\bar{T}\rightarrow 0^{-}$, while $\bar{{\mathbf{u}}}$ is $C^{\infty}$ on $[-2\pi,0]$. The expansion scalar of the physical fluid flow and the Ricci curvature behave as before, namely
\begin{eqnarray}
\mathop {\lim }\limits_{\bar{T} \to 0^{-}}\theta= -\infty, \qquad
\mathop {\lim }\limits_{\bar{T} \to 0^{-}}R_{ab}R^{ab}= \infty,
\label{dustFRW4.eqn}
\end{eqnarray}
and, once again, the proper time along the fluid flow is finite.
The conformal structure and its properties are again completely
analogous to that of the IPS, since this recollapsing FRW universe is certainly proper-time-symmetric with respect to the point of maximum expansion.

The scalars $\bar{\theta}$ and $\bar{C}_{abcd}\bar{C}^{abcd}$ of the unphysical, conformal space-time are found to vanish identically, while $\bar{R}_{ab}\bar{R}^{ab}$ remains constant.

\subsection{The Kantowski and Kantowski-Sachs models}

So far we have only seen cases with a vanishing conformal factor. The following example space-times exhibit a diverging conformal factor at the respective cosmological futures and have previously been shown to admit an IPS \cite{GW1985,Geoffthesis,Scott1998}. These models describe spatially homogeneous but not spatially isotropic cosmologies with an irrotational, geodesic, perfect fluid source, which satisfies a radiation equation of state $p=\frac{1}{3}\mu$ in the Kantowski-Sachs case \cite{KantowskiSachs1966,Wainwright1984} and a dust equation of state $p=0$ in the Kantowski case \cite{Wainwright1984,Kantowski1998}. Both models possess an isometry group $G_4$ acting on spacelike hypersurfaces, which in the case of the Kantowski models contains a three-dimensional subgroup $G_3$ of Bianchi type III acting simply transitively on the hypersurfaces of homogeneity (these cosmologies are therefore also known as locally rotationally symmetric Bianchi III models), while the Kantowski-Sachs models are the only class of spatially homogeneous solutions which do not admit such a $G_3$ \cite{DS,exact}. The irregularity in the conformal structure for these anisotropic cosmologies will express itself in a degenerate conformal metric, thereby confirming the conclusion of the previous section.

The Kantowski and Kantowski-Sachs models satisfy the WCH in the sense that $K$ increases throughout all models.

In comoving coordinates one finds the following form of the metric
\begin{eqnarray}
ds^{2}&=&-Adt^{2}+t\left[A^{-1}dx^{2}+A^{2}b^{-2}\left(dy^{2}+f^{2}dz^{2}\right)\right]\,, \label{kantsachs1.eqn}
\end{eqnarray}
where
\begin{eqnarray}
A(t)&=&1-\frac{4\epsilon b^{2}t}{9}, \ \  t>0, \qquad b=const\neq0, \qquad \mbox{and} \nonumber \\
f(y)&=&\begin{cases} \sin y & \mbox{if $\epsilon =1$ \
(\emph{Kantowski-Sachs}),}
\\
\sinh y &\mbox{if $\epsilon =-1$ \   (\emph{Kantowski}).}
\end{cases}\nonumber
\end{eqnarray}

In order to determine the future behaviour of these models, it is necessary to analyse them separately.

\subsubsection{Indefinite expansion in the Kantowski models}\label{kant}

Choosing $\epsilon =-1$, we find $A=1+(4b^{2}t)/9>0$ for all $t>0$, and $A\rightarrow\infty$ as $t\rightarrow\infty$. It is helpful to pick the cosmic time function $\bar{T}$ in the form
\begin{eqnarray}
\bar{T}=-A^{-1}=-\frac{1}{1+\frac{4b^{2}t}{9}},\label{kant1.eqn}
\end{eqnarray}
since $t=\infty \ \Leftrightarrow \ \bar{T}=0$, and the IPS occurs at $\bar{T}=-1$, i.e., the IPS and the future metric singularity are separated by a unit value difference of $\bar{T}$. Factoring out the divergences yields
\begin{eqnarray}
ds^{2}=\bar{\Omega}^{2}(\bar{T})\left[-\frac{81}{16b^{4}}d\bar{T}^{2}+\frac{9(\bar{T}^{3}+\bar{T}^{2})}{4b^{2}}\{-\bar{T}^{3}dx^{2}+b^{-2}\left(dy^{2}+\sinh^{2}ydz^{2}\right)\}\right],\label{kant2.eqn}
\end{eqnarray}
where the dimensionless conformal factor reads
\begin{eqnarray}
\bar{\Omega}(\bar{T})=(-\bar{T})^{-5/2}.\label{kant3.eqn}
\end{eqnarray}

The conformal metric $d\bar{s}^{2}$ in the square brackets becomes highly degenerate as $\bar{T}\rightarrow 0^{-}$, since all spatial components vanish,
\begin{eqnarray}
d\bar{s}^{2}\rightarrow -\frac{81}{16b^{4}}d\bar{T}^{2}, \qquad
\mbox{as} \qquad \bar{T}\rightarrow 0^{-},
\end{eqnarray}
which, from the previous section, may indicate an asymptotic anisotropy, as $\bar{\Omega}$ is monotonic. Otherwise, $d\bar{s}^{2}$ is $C^{\infty}$ for $\bar{T}\leq 0$.

The conformal factor $\bar{\Omega}(\bar{T})$ diverges as $\bar{T}\rightarrow 0^{-}$, but is positive and $C^{\infty}$ for $\bar{T}<0$. The derivative ratios are determined to satisfy
\begin{eqnarray}
\mathop {\lim }\limits_{\bar{T} \to 0^{-}}\frac{\bar{\Omega}
'}{\bar{\Omega}}=+\infty \qquad \mbox{and} \qquad
\bar{L}\equiv\frac{7}{5}.\label{kant4.eqn}
\end{eqnarray}

The fluid flows of the physical and unphysical space-times are given, respectively, by
\begin{eqnarray}
\mathbf{u}=\frac{4b^{2}(-\bar{T})^{5/2}}{9}\frac{\partial}{\partial\bar{T}},
\qquad
\bar{{\mathbf{u}}}=\frac{4b^{2}}{9}\frac{\partial}{\partial\bar{T}}.
\end{eqnarray}
The physical flow $\mathbf{u}$ vanishes at $\bar{T}=0$ and is otherwise $C^{\infty}$ for $\bar{T}<0$, while the unphysical flow $\bar{{\mathbf{u}}}$ is both regular at, and orthogonal to the slice $\bar{T}=0$.

By (\ref{proptime.eqn}) the proper time of a fluid particle from
the IPS to $\bar{T}=0$ is infinite. The fluid is expanding, i.e.,
$\theta>0$ for $\bar{T}<0$, and the non-zero kinematic quantities
and the curvature invariants satisfy
\begin{eqnarray}
\mathop {\lim }\limits_{\bar{T} \to 0^{-}}\theta=\mathop {\lim }\limits_{\bar{T} \to 0^{-}}\sigma=\mathop {\lim }\limits_{\bar{T} \to 0^{-}}R_{ab}R^{ab}=\mathop {\lim }\limits_{\bar{T} \to 0^{-}}C_{abcd}C^{abcd}=0,
\end{eqnarray}
which clearly shows that the Kantowski models do \emph{not}
possess a big crunch cosmological singularity in the future; in
fact, they correspond to ever-expanding universes. The evolution
of $\theta$ and the Weyl and Ricci curvature are shown in figures
\ref{kant1.img} and \ref{kant2.img}, respectively, for $b=1$.

\begin{figure}[h!]
\centering
\includegraphics[width=0.75\textwidth, height=0.5\textwidth]{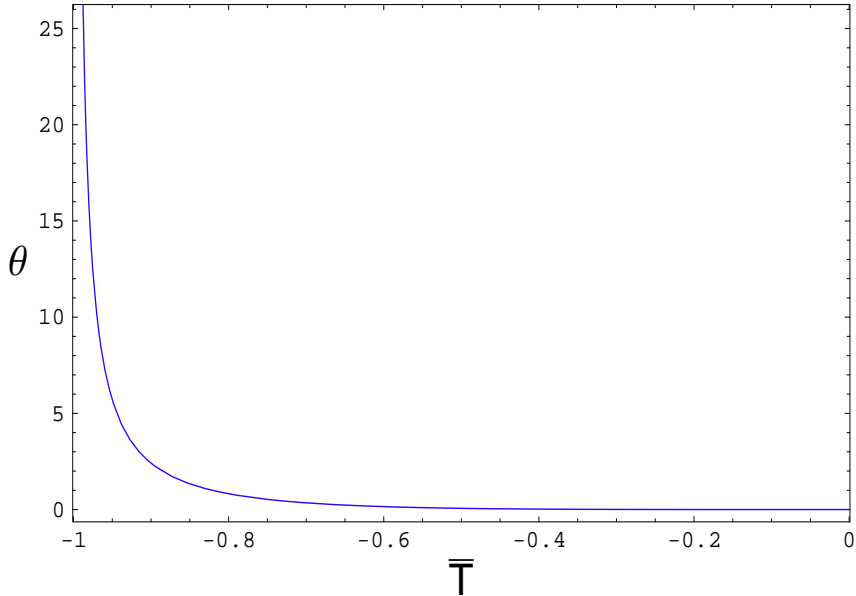}
\caption{{\footnotesize The behaviour of $\theta$ for all of cosmic time, $\bar{T}\in(-1,0)$, in the Kantowski model with $b=1$. }}
\label{kant1.img}
\end{figure}
\begin{figure}[h!]
\centering
\includegraphics[width=0.75\textwidth, height=0.5\textwidth]{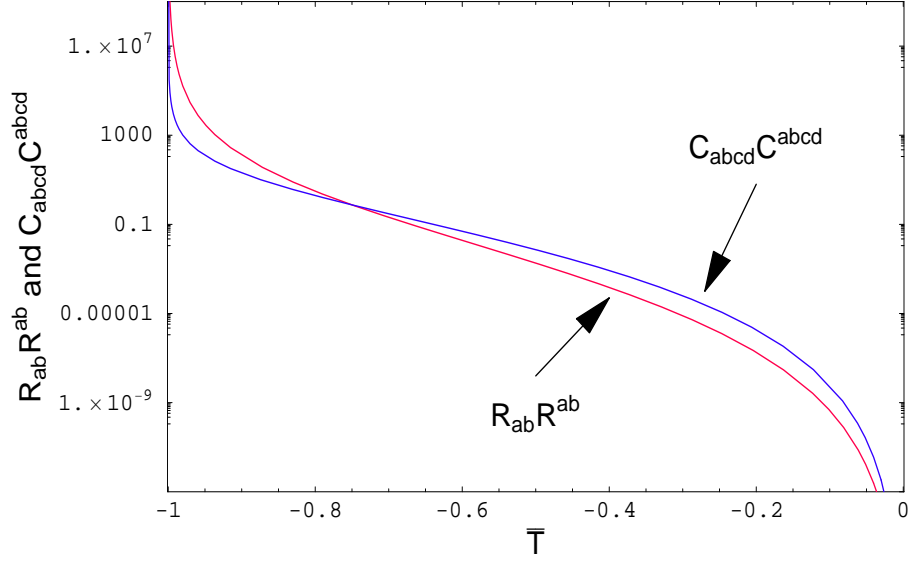}
\caption{{\footnotesize The behaviour of $R_{ab}R^{ab}$ and
$C_{abcd}C^{abcd}$ for all of cosmic time, $\bar{T}\in(-1,0)$, in
the Kantowski model with $b=1$. Note how the Weyl curvature
becomes stronger than the Ricci curvature.}}\label{kant2.img}
\end{figure}
\begin{figure}[h!]
\centering
\includegraphics[width=0.65\textwidth, height=0.45\textwidth]{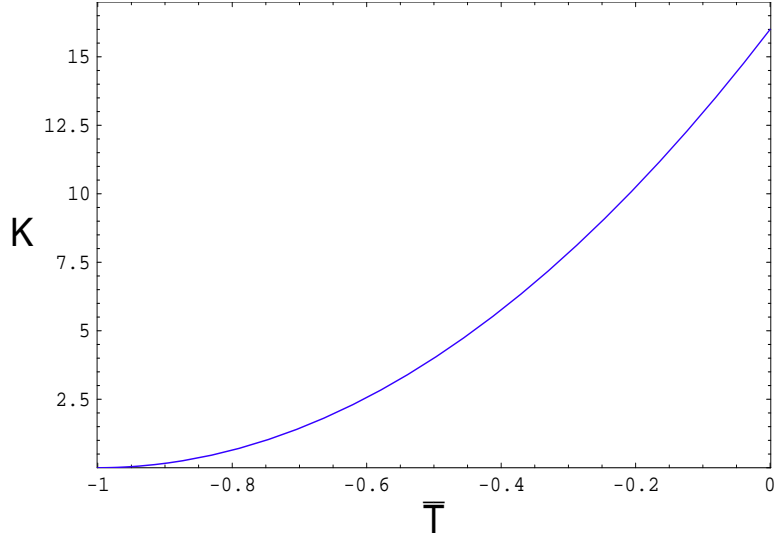}
\caption{{\footnotesize The evolution of $K$ for all of cosmic
time, $\bar{T}\in(-1,0)$, in the Kantowski models ($b$
independent). }} \label{kant3.img}
\end{figure}
The ratio of the non-zero kinematic quantities indeed shows an ``anisotropic" behaviour, since the shear is asymptotically not expansion dominated,
\begin{eqnarray}
\mathop {\lim }\limits_{\bar{T} \to 0^{-}}\frac{\sigma}{\theta}=\sqrt{\frac{1}{6}}.
\end{eqnarray}
The Weyl curvature becomes slightly stronger than the Ricci curvature, in the sense that
\begin{eqnarray}
\mathop {\lim }\limits_{\bar{T} \to 0^{-}}K=16.
\end{eqnarray}
The complete evolution of $K$ (which is independent of $b$) is
graphed in figure \ref{kant3.img}, in strong support of the WCH.

A curvature singularity can, however, now be encountered in the unphysical space-time as $\bar{T}\rightarrow 0^{-}$, since
\begin{eqnarray}
\mathop {\lim }\limits_{\bar{T} \to
0^{-}}\bar{R}_{ab}\bar{R}^{ab}=\mathop {\lim }\limits_{\bar{T} \to
0^{-}}\bar{C}_{abcd}\bar{C}^{abcd}=\infty, \qquad \mbox{and}
\qquad \mathop {\lim }\limits_{\bar{T} \to
0^{-}}\bar{\theta}=-\infty.
\end{eqnarray}

Finally, the determinant  $\bar{\Omega}^{8}\bar{g}$ (where
$\bar{g}=\det\bar{g}_{ab}$) of the physical metric diverges as
$\bar{T}\rightarrow 0^{-}$.

\subsubsection{A future singularity in the Kantowski-Sachs models}\label{kantsachs}\label{kantsachs}

Given   (\ref{kantsachs1.eqn}), it is readily seen that for
$\epsilon =1$ we encounter a future metric singularity when $A=0$,
i.e., when $t\rightarrow t_{s}=\frac{9}{4b^{2}}$. With the
following choice of the cosmic time function $\bar{T}$
\begin{eqnarray}
\bar{T}=-A^{2}=-\left(1-\frac{4b^{2}t}{9}\right)^{2},\label{kantsachs3.eqn}
\end{eqnarray}
which approaches $0$ from below and satisfies $t=t_{s} \
\Leftrightarrow \ \bar{T}=0$, it is evident that this case is
different to the previous example. The IPS occurs at $\bar{T}=-1$
and thus, solving (\ref{proptime.eqn}) shows that this future
metric singularity occurs at finite proper time for the fluid
particles, namely at $\tau =\frac{3}{2b^{2}}$. Rewriting the
metric yields
\begin{eqnarray}
 ds^{2}=\bar{\Omega}^{2}(\bar{T})\left[-\frac{81}{64b^{4}}d\bar{T}^{2}+\left(1-\sqrt{-\bar{T}}\right)\frac{9}{4b^{2}}\left[dx^{2}+\left(-\bar{T}\right)^{3/2}b^{-2}\left(dy^{2}+sin^{2}ydz^{2}\right)\right]\right],
 \label{kantsachs4.eqn}\nonumber\\
\end{eqnarray}
where $\bar{T}\in (-1,0)$, and the dimensionless conformal factor is given by
\begin{eqnarray}
\bar{\Omega}\left(\bar{T}\right)=\frac{1}{(-\bar{T})^{1/4}}.\label{kantsachs5.eqn}
\end{eqnarray}

The $y$- and $z$-components of the conformal metric $d\bar{s}^{2}$ in the square brackets vanish as $\bar{T}\rightarrow 0^{-}$,
\begin{eqnarray}
d\bar{s}^{2}\rightarrow
-\frac{81}{64b^{4}}d\bar{T}^{2}+\frac{9}{4b^{2}}dx^{2}, \qquad
\mbox{as} \qquad \bar{T}\rightarrow 0^{-},
\end{eqnarray}
i.e., $d\bar{s}^{2}$ becomes degenerate as $\bar{T}\rightarrow
0^{-}$, which again may be a structural indication for anisotropy.
(\ref{kantsachs4.eqn}) implies, moreover, that $d\bar{s}^{2}$ is
only $C^{0}$ at $\bar{T}=0$.

The conformal factor $\bar{\Omega}$ diverges as $\bar{T}\rightarrow 0^{-}$, and is positive and $C^{\infty}$ for $\bar{T}<0$. The time derivatives of $\bar{\Omega}$ behave as
\begin{eqnarray}
\mathop {\lim }\limits_{\bar{T} \to 0^{-}}\frac{\bar{\Omega}
'}{\bar{\Omega}}=+\infty, \qquad \mbox{and} \qquad
\bar{L}&=&\frac{\bar{\Omega}
''\bar{\Omega}}{{\bar{\Omega'}}^{2}}=5>1.
\end{eqnarray}

By (\ref{kantsachs4.eqn}), it is apparent that in these
coordinates the timelike fluid flows in the physical and conformal
space-times are given, respectively, by
\begin{eqnarray}
\mathbf{u}=\frac{8b^{2}\left(-\bar{T}\right)^{1/4}}{9}\frac{\partial}{\partial\bar{T}},
\qquad
\bar{{\mathbf{u}}}=\bar{\Omega}\mathbf{u}=\frac{8b^{2}}{9}\frac{\partial}{\partial\bar{T}}.
\end{eqnarray}
The physical fluid flow $\mathbf{u}$ vanishes and becomes $C^{0}$
at the metric singularity. The unphysical fluid flow
$\bar{{\mathbf{u}}}$, however, is $C^{\infty}$, regular at, and
orthogonal to the slice $\bar{T}=0$. The expansion scalar of the
physical fluid flow is plotted in figure \ref{kantsachs0.img} (for
$b=1$)
\begin{figure}[h!]
\centering
\includegraphics[width=0.75\textwidth, height=0.5\textwidth]{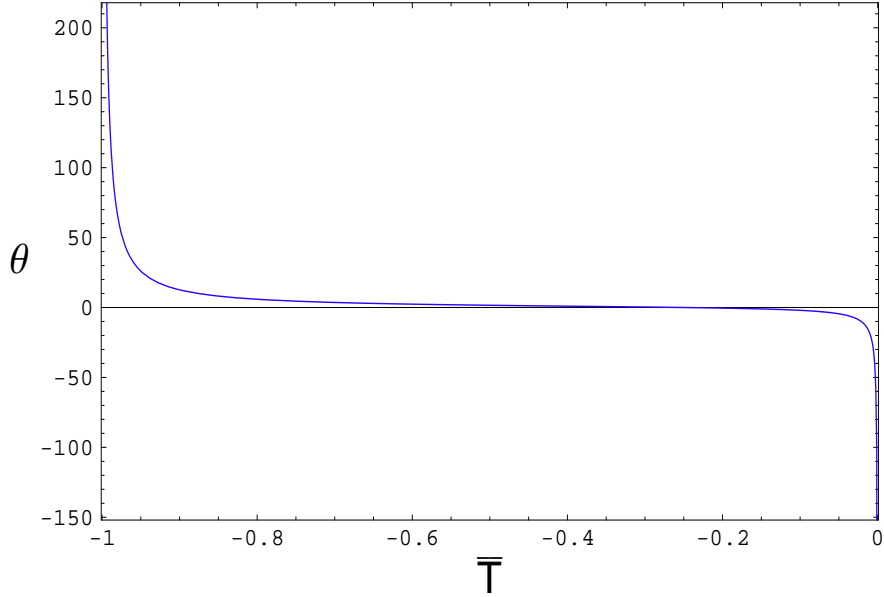}
\caption{{\footnotesize The evolution of $\theta$ between the
singularities in the Kantowski-Sachs model with $b=1$. }}
\label{kantsachs0.img}
\end{figure}
and leads to caustics in the flow lines, while the shear grows unboundedly,
\begin{eqnarray}
\mathop {\lim }\limits_{\bar{T} \to 0^{-}}\theta= -\infty, \qquad
\mathop {\lim }\limits_{\bar{T} \to 0^{-}}\sigma=
\infty.\label{kantsachs6.eqn}
\end{eqnarray}
Their ratio, moreover, shows an asymptotically anisotropic
behaviour
\begin{eqnarray}
\mathop {\lim }\limits_{\bar{T} \to
0^{-}}\frac{\sigma}{\theta}=-\sqrt{\frac{2}{15}}.
\end{eqnarray}

Analysing the Weyl and Ricci curvature invariants, it is found
that
\begin{eqnarray}
\mathop {\lim }\limits_{\bar{T} \to 0^{-}}R_{ab}R^{ab}=\mathop {\lim }\limits_{\bar{T} \to 0^{-}}C_{abcd}C^{abcd}=+\infty.
\end{eqnarray}
The overall behaviour of these curvature invariants between the
singularities is pictured in figure \ref{kantsachs1.img} for
$b=1$. The Ricci curvature diverges at both singularities and it
is readily seen how the Weyl curvature initially decreases and
eventually increases from a finite value to infinity throughout
the evolution of the model.
\begin{figure}[h!]
\centering
\includegraphics[width=0.75\textwidth, height=0.5\textwidth]{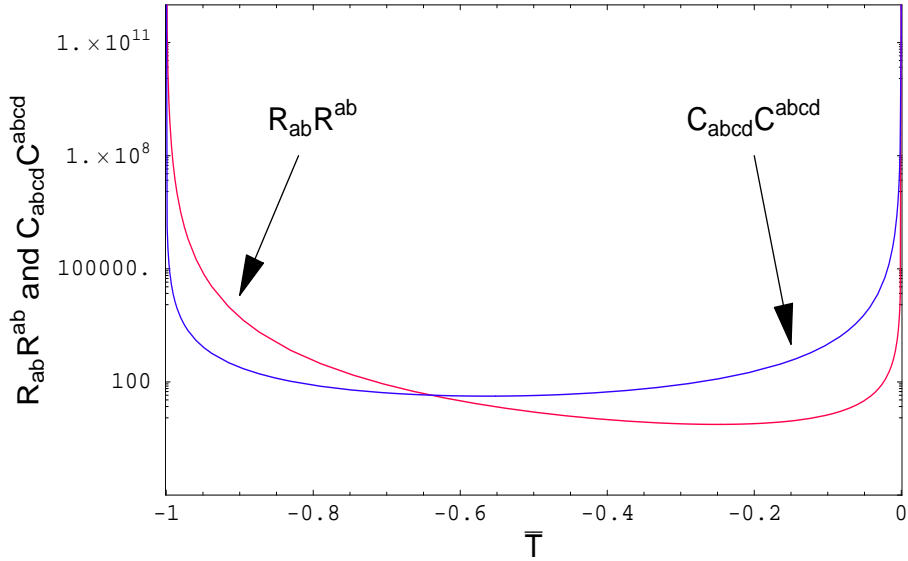}
\caption{{\footnotesize The behaviour of $R_{ab}R^{ab}$ and
$C_{abcd}C^{abcd}$ between the singularities in the
Kantowski-Sachs model with $b=1$. Note how the Weyl curvature
outgrows the Ricci curvature.}} \label{kantsachs1.img}
\end{figure}
Thus, the Kantowski-Sachs models describe recollapsing universes in which the timelike congruence encounters a \emph{physical, scalar polynomial curvature singularity} in finite proper time.

In this class of models, the WCH holds; the scalar $K$ increases
to infinity at the future singularity, $\mathop {\lim
}\limits_{\bar{T} \to 0^{-}}K=+\infty$, as can be seen for the
case $b=1$ in figure \ref{kantsachs2.img}.
\begin{figure}[h!]
\centering
\includegraphics[width=0.75\textwidth, height=0.5\textwidth]{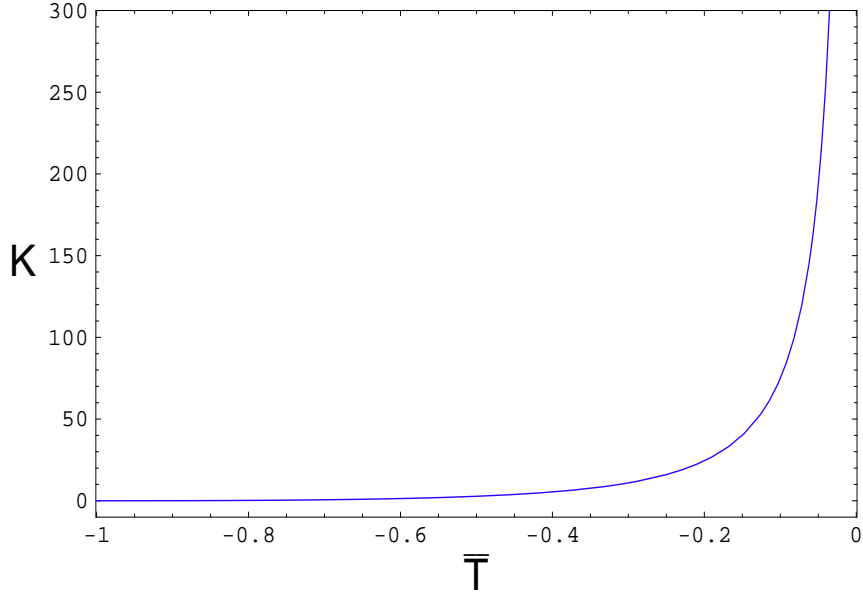}
\caption{{\footnotesize The behaviour of $K$ between the
singularities in the Kantowski-Sachs model with $b=1$. }}
\label{kantsachs2.img}
\end{figure}

Due to the degeneracy of the conformal metric, we again find a scalar polynomial curvature singularity at $\bar{T}=0$ in the unphysical space-time, and an unbounded unphysical expansion
\begin{eqnarray}
\mathop {\lim }\limits_{\bar{T} \to
0^{-}}\bar{R}_{ab}\bar{R}^{ab}=\mathop {\lim }\limits_{\bar{T} \to
0^{-}}\bar{C}_{abcd}\bar{C}^{abcd}=+\infty, \qquad \mathop {\lim
}\limits_{\bar{T} \to 0^{-}}\bar{\theta}=-\infty.
\end{eqnarray}


The essential structural difference to the Kantowski models is that, not only the determinant of the unphysical metric vanishes as $\bar{T}\rightarrow 0^{-}$, but also the determinant of the physical metric,
\begin{eqnarray}
\mathop {\lim }\limits_{\bar{T} \to 0^{-}}g=\mathop {\lim }\limits_{\bar{T} \to 0^{-}}\bar{\Omega}^{8}\bar{g}=0.
\end{eqnarray}
This state of affairs will have an impact on the definitions presented in the next section.

\subsection{Summary of the example cosmologies}

Comparing the four cosmologies discussed in this section provides
further valuable information regarding the possible conformal
structures for an anisotropic future behaviour. In the light of
\emph{quiescent cosmology} and the WCH we would expect
anisotropies to be formed with cosmic evolution by the enhanced
(local) gravitational clumping, at least in the sense that
asymptotically (as $\bar{T}\rightarrow 0^{-}$) the ratios
discussed in (\ref{K.eqn}) and (\ref{kiniso.eqn}) do not vanish
anymore for the non-zero quantities - such that we may talk about
an asymptotic anisotropy. An understanding of the behaviour of the
curvature invariants and the expansion, on the other hand, is
essential for uncovering the structural differences between models
which are ever-expanding and models with recollapsing scenarios.

To facilitate the comparison of the cosmologies we summarise the
behaviour of the characteristic quantities at $\bar{T}=0$ in
tables \ref{examplesphys.tab}, \ref{exampleskin.tab} and
\ref{examplesunphys.tab}. \begin{table}
\caption{\label{examplesphys.tab} Summary of the behaviour of
several quantities as $\bar{T}\rightarrow 0^{-}$ in the physical
space-time of the presented cosmologies. The quantity $\tau_{s}$
denotes proper time elapsed from the initial state to the slice
$\bar{T}=0$ along the fluid flow. (a) means that the relevant
quantity diverges to $+\infty$ at $\bar{T}=0$, but is $C^{\infty}$
for $\bar{T}<0$ away from the IPS, (b) means that the relevant
quantity is finite and $C^{\infty}$ for $\bar{T}\leq0$ away from
the IPS and (c) means that the relevant quantity vanishes as
$\bar{T}\rightarrow0^{-}$.}
\vspace{.4cm}
\begin{tabular}{@{}llllll}
\hline\hline
        Model &$K$&$R_{ab}R^{ab}$&$C_{abcd}C^{abcd}$&$\tau_{s}$&
        IPS\cr  \hline
        rad. FRW ($k=+1$)&  \ 0&  (a)&    \ 0&       finite&      yes
        \cr
    dust FRW ($k=+1$)&    \ 0&   (a)&  \ 0&     finite&    yes \cr
Kantowski&  (b)& (c)& (c)& infinite&   yes\cr
 Kantowski-Sachs&
(a)&  (a)&  (a)&  finite&  yes\cr \hline\hline
    \end{tabular}
\end{table}
\begin{table}
\caption{\label{exampleskin.tab} Behaviour of the fluid flow and
its kinematic quantities as $\bar{T}\rightarrow 0^{-}$ in the
physical space-time, and of the conformal fluid flow and its
expansion, in the investigated cosmologies. (d) means that the
relevant quantity diverges to $-\infty$ at $\bar{T}=0$, but is
$C^{\infty}$ for $\bar{T}<0$ away from the IPS (see table
\ref{examplesphys.tab} for an explanation of the other
abbreviations).}
\vspace{.4cm}
\begin{tabular}{@{}lllllllll}
\hline\hline
       Model
        &$\mathbf{u}$&$\theta$&$\dot{u}^{a}$&$\sigma$&$\omega$&$\sigma/\theta$&$\mathbf{\bar{u}}$&$\bar{\theta}$\cr
        \hline
        rad. FRW ($k=+1$)&  (a)&    (d)&   \ 0&       \ 0&  0&   \ 0&   (b)&
        (c)\cr
    dust FRW ($k=+1$)&   (a)&   (d)&  \ 0&     \ 0&  0&  \ 0&  (b)&  \ 0
    \cr
Kantowski&  (c)& (c)& \ 0&  (c)&   0&    (b)&  (b)&  (d)\cr
Kantowski-Sachs&  (c)&  (d)&  \ 0&  (a)& 0&   (b)&  (b)& (d)\cr
\hline\hline
    \end{tabular}
\end{table}
\begin{table}
\caption{\label{examplesunphys.tab} Some characteristic properties
of the conformal structure in the investigated cosmologies as
$\bar{T}\rightarrow 0^{-}$. $C_{0}^{n}$ denotes the degree of
differentiability of $\mathbf{\bar{g}}$ at $\bar{T}=0$. (e) means
that some spatial components of the metric vanish as
$\bar{T}\rightarrow 0^{-}$. See table \ref{examplesphys.tab} for
the explanation of the other abbreviations.}
\vspace{.4cm}
\begin{tabular}{@{}lllllllll} \hline\hline
       Model &degenerate
        $\mathbf{\bar{g}}$&$C_{0}^{n}$&$\bar{\Omega}$&$\bar{L}$&$\bar{\Omega}^{8}|\bar{g}|$&$\bar{C}_{abcd}\bar{C}^{abcd}$&$\bar{R}_{ab}\bar{R}^{ab}$\cr
        \hline
        rad. FRW ($k=+1$)&  no& $n=\infty$& (c)&   \ 0&    (c)&    \ 0& (b)
        \cr
    dust FRW ($k=+1$)&   no&   $n=\infty$&  (c)&   $<1$&  (c)&    \ 0&  (b)
    \cr
Kantowski&  (e)& $n=\infty$& (a)&  $>1$& (a)&    (a)&  (a)\cr
Kantowski-Sachs&  (e)&  $n=0$&  (a)&  $>1$& (c)&  (a)&  (a)\cr \hline\hline
    \end{tabular}
\end{table}From the behaviours of $K$ in table
\ref{examplesphys.tab}, of the fluid flow and the asymptotic ratio
of the non-zero kinematic quantities in table
\ref{exampleskin.tab} and of the conformal factor and conformal
metric in table \ref{examplesunphys.tab}, we expect that a
conformal structure which leads to an anisotropic future behaviour
should satisfy the following conditions at $\bar{T}=0$:
\begin{enumerate}
\item the conformal factor diverges and $\mathop {\lim }\limits_{\bar{T} \to 0^{-}}\bar{L}=\bar{\lambda}\geq1$,
\item the conformal metric becomes degenerate with some spatial components vanishing,
\item the physical fluid flow $\mathbf{u}$ vanishes, and
\item the unphysical fluid flow $\mathbf{\bar{u}}$ remains at least $C^{2}$.
\end{enumerate}
For the first point, the fact that $\bar{\lambda}\geq 1$ for a
diverging $\bar{\Omega}$ will be proven in Lemma
\ref{conffac1.lem} of \ref{conffac}. For the time being it still
remains unclear, however, whether a degenerate conformal metric
and a vanishing conformal factor can also lead to an anisotropic
scenario. The second point is already justified by the results of
section \ref{confiso} in the sense that an irregular conformal
metric is necessary in a conformal structure with a conformal
factor $\bar{\Omega}$ which is well-behaved for $\bar{T}\in(-c,0)$
(for some $c>0$), if an asymptotic anisotropy is to occur.

From this state of affairs we can already infer that it will not
be possible to construct a similar condition to requirement (ii)
of Definition \ref{ISdef1.def} in our new framework, i.e., we lose
part of the regularity in the unphysical space-time. As can be
seen from tables \ref{exampleskin.tab} and
\ref{examplesunphys.tab}, the degeneracy in the metric apparently
leads to a scalar polynomial curvature singularity and a diverging
$\bar{\theta}$ at $\bar{T}=0$ for the unphysical space-time.  This
complicates investigations significantly and will be discussed in
an upcoming paper \cite{HoehnScott2}. Thus, we cannot require that
the physical manifold be a submanifold of the conformal space-time
for the cosmological future.

The third and fourth points concerning the fluid flow will be essential, as in the case of the IPS, and will be employed in results concerning the expansion scalar, which will appear in a forthcoming paper in which further example cosmologies will be discussed \cite{HoehnScott2}.

In order not to be too restrictive on the conformal metric we
will, furthermore, take into account that it may become only
$C^{0}$ at $\bar{T}=0$, in accordance with table
\ref{examplesunphys.tab}. Finally, a distinguishing characteristic
between models which possess a future singularity, and those which
describe ever-expanding universes, is the determinant of the
physical metric, $\bar{\Omega}^{8}\bar{g}$, which vanishes in the
former cosmologies and diverges in the latter cosmologies. This
will be of importance in results concerning the notion of strong
curvature in the new frameworks \cite{HoehnScott2}.

We emphasise that, irrespective of regularity or irregularity, the above conditions leave (as explicity seen in the example of section \ref{radFRW}) a considerable freedom in choosing the conformal factor.

\section{New definitions for cosmological futures}\label{newdefs}

Motivated by the previous two sections, we will now proceed to give the new definitions for conformal structures for cosmological futures. In fact, in the interests of completeness, we will start with definitions for isotropic future behaviour in subsection \ref{isotropic} prior to presenting the physically more interesting conformal structures with an asymptotically anisotropic future evolution in subsection \ref{anisotropic}.

Unlike in the case of Definition \ref{ISdef1.def}, we will require
that the conformal factor and the metric be at least $C^{2}$
(instead of $C^{3}$), in agreement with the common definition of a
space-time in which the physical metric is at least $C^{2}$. Some
technical details concerning the conformal factor are clarified in
\ref{conffac}, such as, for example, why in some cases we require
$\bar{\lambda}\neq1,2$ for which a general conclusion about the
behaviour of some functions of $\bar{\Omega}$ and its derivatives
is not evident.

As before, we will denote some relevant quantities with a $\
\bar{} \ $, in order to emphasise that we are now dealing with
conformal structures for cosmological futures.

\subsection{New definitions for isotropic future behaviour}\label{isotropic}

The two FRW models exhibited a cosmological future structure and physical behaviour which is essentially the time-reverse of an IPS. In close analogy to Definition \ref{ISdef1.def}, we provide the following definition:

\begin{defn}[\bf{Isotropic future singularity (IFS)}]\label{IFS.def}{A space-time $(\mathcal{M},\mathbf{g})$ is said to admit an \emph{isotropic future singularity} if there exists a space-time $(\bar{\mathcal{M}},\mathbf{\bar{g}})$, a smooth cosmic time function $\bar{T}$ defined on $\bar{\mathcal{M}}$, and a conformal factor $\bar{\Omega}\left(\bar{T}\right)$ which satisfy
\begin{enumerate}
\item $\mathcal{M}$ is the open submanifold $\bar{T}<0$,
\item $\mathbf{g}=\bar{\Omega}^{2}\left(\bar{T}\right)\mathbf{\bar{g}}$ on $\mathcal{M}$, with $\mathbf{\bar{g}}$ regular (at least $C^{2}$ and non-degenerate) on an open neighbourhood of $\bar{T}=0$,
\item $\bar{\Omega}\left(0\right)=0$ and $\exists \ c>0$ such that $\bar{\Omega}\in C^{0}[-c,0]\cap C^{2}[-c,0)$ and $\bar{\Omega}$ is positive on $[-c,0)$,
\item $\bar{\lambda}\equiv \mathop {\lim }\limits_{\bar{T} \to 0^{-}}\bar{L}\left(\bar{T}\right)$ exists, $\bar{\lambda}\neq 1$, where $\bar{L}\equiv \frac{\bar{\Omega} ''}{\bar{\Omega}}\left(\frac{\bar{\Omega}}{\bar{\Omega} '}\right)^{2}$ and a prime denotes differentiation with respect to $\bar{T}$.
\end{enumerate}}
\end{defn}

It should be noted that if the model admits an IPS as well, the above conformal relation will not, in general, be the same as that for the IPS.

Additionally, in analogy to Definition \ref{fluidflow1.def}, we require:

\begin{defn}[\bf{IFS fluid congruence}]\label{IFSfluid.def}{With any unit timelike congruence $\mathbf{u}$ in $\mathcal{M}$ we can associate a unit timelike congruence $\mathbf{\bar{u}}$ in $\bar{\mathcal{M}}$ such that
\begin{eqnarray}
\mathbf{\bar{u}}=\bar{\Omega}\mathbf{u} \qquad \mbox{in} \qquad
\mathcal{M}.\nonumber
\end{eqnarray}
\begin{description}
\item[(a)] If we can choose $\mathbf{\bar{u}}$ to be regular (at least $C^{2}$) on an open neighbourhood of $\bar{T}=0$ in $\bar{\mathcal{M}}$, we say that $\mathbf{u}$ is \emph{regular at the IFS.}
\item[(b)] If, in addition, $\mathbf{\bar{u}}$ is orthogonal to $\bar{T}=0$, we say that $\mathbf{u}$ is \emph{orthogonal to the IFS.}
\end{description}}
\end{defn}

Based on the results in section \ref{confiso} there is a further possibility for a conformal structure with an isotropic future behaviour, which does not necessarily lead to a future singularity. Some open FRW universes, for example, might satisfy these conditions.

\begin{defn}[\bf{Future isotropic universe (FIU)}]\label{FIU.def}{A space-time $(\mathcal{M},\mathbf{g})$ is said to be a \emph{future isotropic universe} if there exists a space-time $(\bar{\mathcal{M}},\mathbf{\bar{g}})$, a smooth cosmic time function $\bar{T}$ defined on $\bar{\mathcal{M}}$, and a conformal factor $\bar{\Omega}\left(\bar{T}\right)$ which satisfy
\begin{enumerate}
\item $\mathop {\lim }\limits_{\bar{T} \to 0^{-}}\bar{\Omega}\left(\bar{T}\right)=+\infty$ and $\exists \ c>0$ such that $\bar{\Omega}\in C^{2}[-c,0)$ and $\bar{\Omega}$ is strictly monotonically increasing and positive on $[-c,0)$,
\item $\bar{\lambda}$ as defined above exists, $\bar{\lambda}\neq 1,2$, and $\bar{L}$ is continuous on $[-c,0)$, and
\item otherwise the conditions of Definitions \ref{IFS.def} and \ref{IFSfluid.def} are fulfilled.
\end{enumerate}}
\end{defn}

\subsection{Anisotropic future endless universes and\\anisotropic future singularities}\label{anisotropic}

In an attempt to complete the framework of an IPS and the
formalisation of \emph{quiescent cosmology} and the WCH we now
present new geometric definitions which describe a non-isotropic
future evolution of the universe. Since it is not yet clear
whether our universe will expand forever or recollapse in finite
proper time, it seems reasonable, in fact, to give two
definitions, one for each scenario.

Before we present the definitions for the conformal structures, we need to define the notion of a limiting causal future, which will be essential in defining where $\bar{T}$ becomes zero, since a cosmic time function cannot exist where the metric becomes degenerate (the existence of a cosmic time function requires stable causality). This limiting causal future is required to be a subset of a larger manifold of which the physical manifold is a submanifold.

\begin{defn}[\bf{Limiting causal future}]\label{causalfuture.def}{Let $(\mathcal{M},\g)$ be a space-time, where $\mathcal{M}\subset\bar{\mathcal{M}}$. We define the \emph{limiting causal future} of $\mathcal{M}$, denoted $F^{+}(\mathcal{M})$, as follows}
\begin{eqnarray}
F^{+}(\mathcal{M}):=\{p\in\bar{\mathcal{M}}\left.\right|\ \exists
\ \mbox{a future inextendible causal curve} \
\gamma_{p}(s):[0,a)\rightarrow\mathcal{M} \ \mbox{in} \
\mathcal{M},\nonumber\\ \mbox{where} \
a\in\mathbb{R}^{+}\cup\{\infty\}, \ \mbox{such that, in} \
\bar{\mathcal{M}}, \ p=\gamma_{p}(a)\equiv\mathop {\lim
}\limits_{s \to a}\gamma_{p}(s)\}. \nonumber
\end{eqnarray}
\end{defn}

This, furthermore, enables us to define the following type of
degeneracy which is necessary in defining ``how degenerate" the
conformal metric becomes in the frameworks presented below.

\begin{defn}[\bf{Causal degeneracy}]\label{causaldeg.def}{Consider $p\in F^{+}(\mathcal{M})$. Let $\gamma_{p}(s)$ be a causal curve in $\mathcal{M}$ as defined above with a limiting tangent vector $\gamma_{p} '\neq 0$ at $p$. The metric $\g$ is said to be \emph{causally degenerate} at $p$ if there exists such a curve $\gamma_{p}$ which satisfies $\g(\gamma_{p} '(a),X)=0$ $\forall \ X\in T_{p}\bar{\mathcal{M}}$. (Note that this assumes that the metric is continuous on an open neighbourhood of $p$.)}
\end{defn}

We will henceforth require that $F^{+}(\mathcal{M})$ be a non-empty set, and, for simplicity, that $F^{+}(\mathcal{M})$ corresponds solely to the final cosmological state. The function $\bar{T}$ is defined to be $0$ on $F^{+}(\mathcal{M})$.

\begin{defn}[\bf{Anisotropic future endless universe (AFEU)}]\label{AFEU.def}{A space-time $(\mathcal{M},\g)$ is said to be an \emph{anisotropic future endless universe} if there exists a larger manifold $\bar{\mathcal{M}}\supset\mathcal{M}$, a space-time $(\mathcal{M},\mathbf{\bar{g}})$, a smooth function $\bar{T}$ defined on $\mathcal{M}\cup F^{+}(\mathcal{M})$, where $F^{+}(\mathcal{M})\neq\emptyset$, and a conformal factor $\bar{\Omega}(\bar{T})$ which satisfy
\begin{enumerate}
\item $\bar{T}=0$ on $F^{+}(\mathcal{M})$, and $\bar{T}$ is a cosmic time function on $\mathcal{M}$ with range $\bar{T}<0$,
\item $\mathbf{g}=\bar{\Omega}^{2}\left(\bar{T}\right)\mathbf{\bar{g}}$ on $\mathcal{M}$, and $\mathbf{\bar{g}}$ is $C^{0}$ and degenerate, but \emph{not causally degenerate}, on $F^{+}(\mathcal{M})$,
\item $\mathop {\lim }\limits_{\bar{T} \to 0^{-}}\bar{\Omega}\left(\bar{T}\right)=+\infty$, and $\exists \ c>0$ such that $\bar{\Omega}\in C^{2}[-c,0)$ and $\bar{\Omega}$ is strictly monotonically increasing and positive on $[-c,0)$,
\item $\bar{\lambda}$ as defined above exists, $\bar{\lambda}\neq 1$ and $\bar{L}$ is continuous on $[-c,0)$, and
\item $\mathop {\lim }\limits_{\bar{T} \to 0^{-}}\bar{\Omega}^{6}|\bar{g}|=+\infty$ across all of $F^{+}(\mathcal{M})$\footnote{The reason for this more stringent choice, $\mathop {\lim }\limits_{\bar{T} \to 0^{-}}\bar{\Omega}^{6}|\bar{g}|=+\infty$, rather than $\mathop {\lim }\limits_{\bar{T} \to 0^{-}}\bar{\Omega}^{8}|\bar{g}|=+\infty$, is the expansion $\theta$, which in the latter case can possibly become negative along the flow lines prior to $\bar{T}=0$ \cite{HoehnScott2}.}, where $\bar{g}$ is the determinant of $\mathbf{\bar{g}}$.
\end{enumerate}}
\end{defn}

Note that since $(\mathcal{M},\mathbf{\bar{g}})$ is a space-time by construction, $\mathbf{\bar{g}}$ is automatically regular (at least $C^{2}$ and non-degenerate) on $\mathcal{M}$. Unlike in Definition \ref{ISdef1.def}, the physical space-time and the conformal space-time may possess the same manifold $\mathcal{M}$.

Motivated by the behaviour of the fluid flow in the Kantowski and Kantowski-Sachs cosmologies, we additionally require the following condition on this slice, in order to guarantee an appropriate behaviour of the fluid flow quantities:

\begin{defn}[\bf{AFEU fluid congruence}]\label{AFEUfluid.def}{With any unit timelike congruence $\mathbf{u}$ in $(\mathcal{M},\g)$ we can associate a unit timelike congruence $\mathbf{\bar{u}}$ in $(\mathcal{M},\mathbf{\bar{g}})$ such that
\begin{eqnarray}
\mathbf{\bar{u}}=\bar{\Omega}\mathbf{u} \qquad \mbox{in} \qquad
\mathcal{M}.\nonumber
\end{eqnarray}
\begin{description}
\item[(a)] If we can further choose $\mathbf{\bar{u}}$ to be regular (at least $C^{2}$) on $\mathcal{M}\cup F^{+}(\mathcal{M})$, we say that $\mathbf{u}$ is \emph{regular on the slice $\bar{T}=0$.}
\item[(b)] If, in addition, $\mathbf{\bar{u}}$ is orthogonal to $\bar{T}=0$, we say that $\mathbf{u}$ is \emph{orthogonal to the slice $\bar{T}=0$.}
\end{description}}
\end{defn}

Case (b) is possible since $\mathbf{\bar{g}}$ is not causally
degenerate at $\bar{T}=0$. The definition of an AFEU is
schematically represented in figure \ref{Afeu.img}.
\begin{figure}[h!] \centering
\includegraphics[width=0.75\textwidth, height=0.55\textwidth]{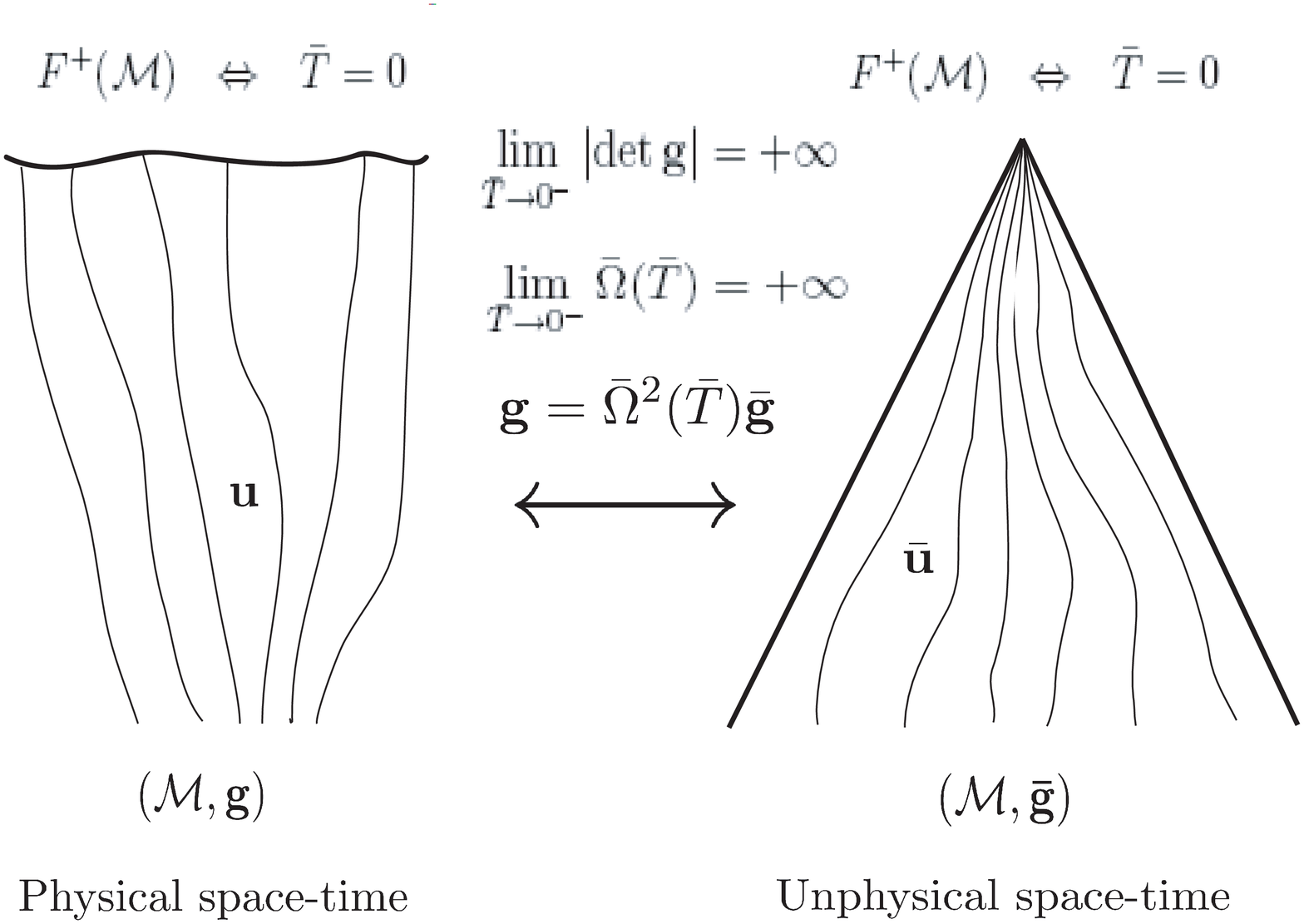}
\caption{{\footnotesize Pictorial interpretation of the definition of an AFEU. The fact that there now exists a singularity for the unphysical space-time, but not for the physical space-time, should not be interpreted as the conformal space-time possessing a more complicated structure. In fact, due to the continuity of $\mathbf{\bar{g}}$ the conformal space-time is, in a sense, more regular than the physical space-time, which is important for analytical considerations.}}
\label{Afeu.img}
\end{figure}

Finally, for an asymptotically anisotropic recollapsing universe we define:

\begin{defn}[\bf{Anisotropic future singularity (AFS)}]\label{AFS.def}{A space-time $(\mathcal{M},\g)$ is said to admit an \emph{anisotropic future singularity} if there exists a larger manifold $\bar{\mathcal{M}}\supset\mathcal{M}$, a space-time $(\mathcal{M},\mathbf{\bar{g}})$, a smooth function $\bar{T}$ defined on $\mathcal{M}\cup F^{+}(\mathcal{M})$, where $F^{+}(\mathcal{M})\neq\emptyset$, and a conformal factor $\bar{\Omega}(\bar{T})$ which satisfy
\begin{enumerate}
\item conditions 1. - 4. of Definition \ref{AFEU.def}, and
\item $\mathop {\lim }\limits_{\bar{T} \to 0^{-}}\bar{\Omega}^{8}|\bar{g}|=0$ across all of $F^{+}(\mathcal{M})$, where $\bar{g}$ is the determinant of $\mathbf{\bar{g}}$.
\end{enumerate}}
\end{defn}

Once more, we additionally define the following for the fluid flow:

\begin{defn}[\bf{AFS fluid congruence}]\label{AFSfluid.def}{With any unit timelike congruence $\mathbf{u}$ in $(\mathcal{M},\g)$ we can associate a unit timelike congruence $\mathbf{\bar{u}}$ in $(\mathcal{M},\mathbf{\bar{g}})$ such that
\begin{eqnarray}
\mathbf{\bar{u}}=\bar{\Omega}\mathbf{u} \qquad \mbox{in} \qquad
\mathcal{M}. \nonumber
\end{eqnarray}
\begin{description}
\item[(a)] If we can further choose $\mathbf{\bar{u}}$ to be regular (at least $C^{2}$) on $\mathcal{M}\cup F^{+}(\mathcal{M})$, we say that $\mathbf{u}$ is \emph{regular at the AFS.}
\item[(b)] If, in addition, $\mathbf{\bar{u}}$ is orthogonal to $\bar{T}=0$, we say that $\mathbf{u}$ is \emph{orthogonal to the AFS.}
\end{description}}
\end{defn}

The definition of an AFS is schematically represented in figure
\ref{Afs.img}. \begin{figure}[h!] \centering
\includegraphics[width=0.75\textwidth, height=0.55\textwidth]{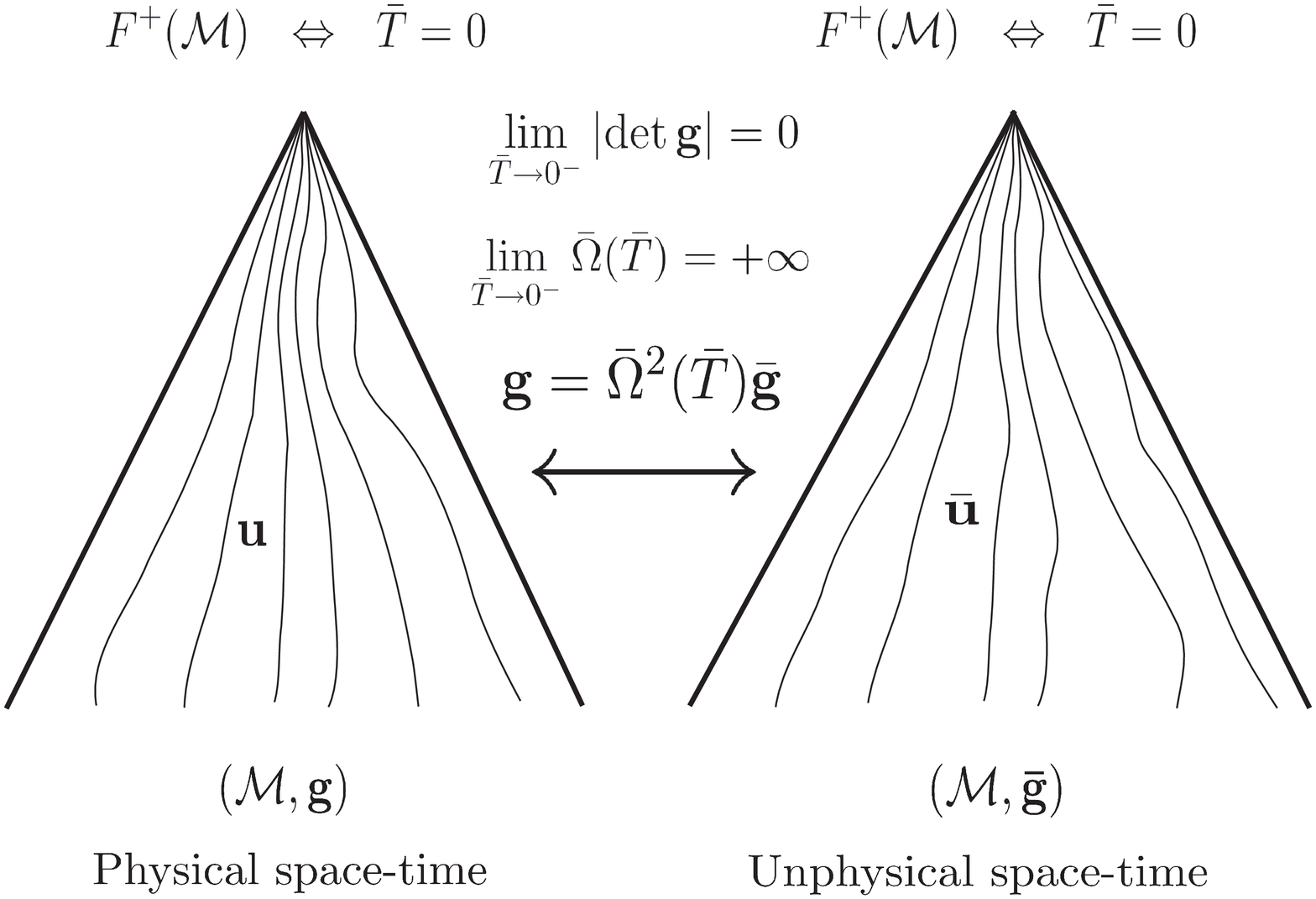}
\caption{{\footnotesize Pictorial interpretation of the definition of an AFS. The fact that the schematic representation of both space-times appears similar should not be interpreted as both space-times possessing a similar structure. Due to the continuity of $\mathbf{\bar{g}}$ the conformal space-time is slightly more regular than the physical space-time, which is important for analytical considerations.}}
\label{Afs.img}
\end{figure}

\begin{rem}{For simplicity, and without loss of generality, we will henceforth refer to $F^{+}(\mathcal{M})$ as $\bar{T}=0$. As in the case of the IPS, we will refer to $(\mathcal{M},\mathbf{g})$ as the \emph{physical space-time}, while the conformal space-time $(\bar{\mathcal{M}},\mathbf{\bar{g}})$ ($(\mathcal{M},\mathbf{\bar{g}})$ respectively) will be referred to as the \emph{unphysical space-time}.
}
\end{rem}

Some physical implications of the definitions of an AFEU and an AFS, concerning the expansion $\theta$ and the notion of strong curvature in these frameworks, will appear in the forthcoming article \cite{HoehnScott2}.

\subsection{Remark concerning the conformal factor}

At this point it should be emphasised that the conformal factor plays a prominent role in all the new frameworks, by absorbing the divergent behaviour of the physical metric in the definitions of an AFEU and an AFS and the entire singular behaviour in the case of the IFS and FIU. Its asymptotic differentiability and monotonicity properties are, furthermore, of key importance in the derivation of implications of all the new definitions, for which \ref{conffac} provides essential technical properties. In fact, the conditions on the conformal factor in the respective definitions do not uniquely define it, but rather leave a considerable freedom in choosing it, as explicitly seen in the example demonstrated in section \ref{radFRW}. For the case of the IPS it was shown that there exists an entire class of conformal factors, related by rescalings \cite{Scott2002}. Similar results should hold in the present context.

\section{Discussion and Outlook}\label{discussion}

In this article we have argued that the definition of an IPS (see
Definition \ref{ISdef1.def}) does not guarantee an anisotropic
cosmological future as one would desire in the light of
\emph{quiescent cosmology} and the WCH. It is therefore necessary
to devise a complementary framework to complete the formalisation
of these two schools of thought.

In order to elaborate an appropriate new framework in terms of
analogous conformal structures we have firstly analysed the
relation between isotropy and conformal structures. In section
\ref{confiso} we could show that conformal structures, whose
conformal factor is merely a function of a cosmic time function
$T$ and which shows a monotonic behaviour near the cosmological
past or future, $T=0$, necessarily lead to an asymptotic isotropy
as discussed in (\ref{K.eqn}) and (\ref{kiniso.eqn}), if a
conformal metric is employed which remains regular at $T=0$.
Interestingly, this result holds for both a vanishing and a
diverging conformal factor at $T=0$ and a cosmic time function
which approaches $0$ from above (initial scenario) or below (final
scenario). The upshot of this state of affairs is, that \emph{if
we are to employ conformal structures in cosmology with a
monotonic conformal factor as a function of a cosmic time, then the choice of a conformal structure with 1.\ a regular conformal
metric necessarily produces an asymptotically isotropic past or future scenario, and 2.\ an irregular conformal metric is required,
if the past or future scenario is supposed to be asymptotically
anisotropic}.

For further motivation and guidance in constructing appropriate
conformal structures for cosmological futures, we have analysed a number of example cosmologies of which four are displayed in section \ref{examples}, namely two
closed FRW universes, the Kantowski models as examples for
anisotropic ever-expanding universes and the Kantowski-Sachs
models as examples for anisotropic recollapsing universes.

For technical interest, we have, moreover, in \ref{conffac},
analysed the behaviour of a monotonic conformal factor, which is a
function of a cosmic time and which may diverge or vanish at the
cosmological future or past, as well as of some functions of it
and its derivatives. These results are essential for the
derivation of the implications of the new frameworks.

The collected information has led us to define new conformal
structures with isotropic future evolution, namely the
\emph{Isotropic Future Singularity} (IFS) and the \emph{Future
Isotropic Universe} (FIU), and, most importantly, we defined new
conformal structures with degenerate conformal metrics and
presumably anisotropic future evolution, namely the
\emph{Anisotropic Future Endless Universe} (AFEU) and the
\emph{Anisotropic Future Singularity} (AFS). The structures of the
latter two definitions differ significantly from the isotropic
cases due to the inherent irregularities. The choice of a
degenerate conformal metric is emphasised by the results of
section \ref{confiso} and the explored example cosmologies. As regards comparisons to other approaches, it should be stressed that the new geometric frameworks are independent of the matter content. Furthermore, we have thus far not made explicit use of the EFE which suggests possible applicability in modified and extended gravity theories as well.

The two closed FRW universes serve as concrete example models which admit
an IFS and the Kantowski and Kantowski-Sachs models serve as AFEU
and AFS example cosmologies, respectively. The fact that the
latter two models do not possess a conformal metric which becomes
\emph{causally degenerate} at $\bar{T}=0$ can most readily be seen
in the diagonal form of the metric. The $\bar{g}_{00}$ components
of their (continuous) conformal metrics do not vanish, hence
$\mathbf{\bar{g}}$ does not become \emph{causally degenerate} at
$\bar{T}=0$. The structures of further example cosmologies which influenced and support the shape of the new frameworks will be presented in the forthcoming paper and, along with the ones presented here, concretely show that there exist a number of explicit models which do indeed satisfy the new definitions (including the FIU definition, as will be seen in \cite{HoehnScott2}).

The isotropic physical implications of the IFS and FIU definitions
are essentially given by the results in section \ref{confiso},
which also show that the IFS conditions necessarily lead to a
future space-time singularity, while the conditions of the FIU
admit a variety of scenarios which are not, in general, singular.

Deriving the physical implications of the AFEU and the AFS
definitions is significantly more complicated as a consequence of
the irregularities in the unphysical space-time. The results of
section \ref{confiso}, however, indicate that the structure of the
definitions cannot be any simpler if conformal structures are to
be maintained. Nevertheless, due to the inherent continuity of the
conformal metric and the regularity of the reference fluid congruence in the unphysical space-time as $\bar{T}\rightarrow 0^-$, the conformal space-time is more
regular than the physical space-time, which is very important
analytically, allowing a number of physical properties regarding the associated final states of the universe to be derived, as will be seen in our forthcoming article
\cite{HoehnScott2}. As in the case of the IPS, where the physical metric becomes degenerate, the limits along the flow lines of the various scalar quantities constructed from curvature tensors and kinematic quantities have meaning as the cosmic time function approaches zero from below, irrespective of whether or not, due to the degeneracy of the metric, the scalar quantities concerned actually exist at $\bar{T}=0$. In view of the above, the requirement that the conformal metric be degenerate thus does not render the definitions of an AFEU and an AFS dysfunctional, in spite of the departure from the high degree of regularity built into the conformal techniques originally developed by Penrose in the differential-topological context for characterisation of asymptotically-flat space-times and which were subsequently utilised in the IPS-framework. The lesser degree of regularity required by the new definitions is indeed sufficient to prove important physical results, relating to the final states of the universe encoded by these definitions which can therefore be employed in a useful way in the geometric discussion of asymptotic states in a class of cosmologies. 

In particular, by making use of the regularity of the unphysical reference congruence and the degeneracy, it is possible to derive matter independent physical results concerning the behaviour of the expansion $\theta$ at late times in the AFEU/AFS-framework. Furthermore, the notion of strong curvature will be analysed and using the continuity and degeneracy of the new definitions it will be shown that, for a wide class of cosmologies admitting an AFS, the singularity turns out to be a \emph{deformationally strong singularity} for timelike geodesics, which provides some information about curvature and the geometric structure \cite{HoehnScott2}. It also underlines the final physical character of the AFS. 

We emphasise that, along with the degeneracy of the conformal metric and the regularity of the unphysical fluid flow, the conformal factor plays a key role in the constructed frameworks. The results and conclusions of section \ref{confiso} are based on its ``regular'' asymptotic differentiability and monotonicity properties and do not hold under different circumstances. Utilising a conformal factor which satisfies different conditions, such as, e.g., an oscillatory behaviour, might significantly alter conclusions and the physical properties of an associated conformal framework. Additionally, we have already seen, that the requirements on the conformal factor, provided by the definitions, do not uniquely define it, but rather leave a considerable freedom in choosing it, which, in view of constructing a conformal structure for a specific cosmology, we regard as an advantage. 

In the light of \emph{quiescent cosmology} and the WCH we are more
interested in the definitions of an AFEU and an AFS. Given that we have not made explicit use of an initial state in the construction of these definitions, we note that
they should, in principal, also be compatible with inflationary
cosmologies which evolve into asymptotically anisotropic
cosmological futures.

The examination of the precise degree of anisotropy in the AFS and the
AFEU, still remains to be finalised. Thus far, the analysed
example cosmologies satisfying these definitions merely indicate
that the AFS and AFEU admit a great variety of anisotropic future
evolutions \cite{HoehnScott2}.

At this point we would like to comment briefly on the relationship of the present framework to another approach, which has been in vogue during the last years, namely the \emph{dynamical systems approach} \cite{DS,UEWE,LEUW,RU05}. Both approaches are based on conformal rescalings and exhibit a certain ``duality'' in their features, as we will point out in the sequel. While the choice of a conformal factor in the AFEU/AFS-framework is of a rather mathematical nature with the aim to render the conformal metric non-singular (although degenerate) and thereby does not usually equip the conformal factor with a particular dimension (the choices in the Kantowski and Kantowski-Sachs models were dimensionless), the motivation for the conformal factor in the dynamical systems approach has a physical basis. In order to regularise the EFE in the vicinity of singularities via conformal rescalings, the Raychaudhuri equation and singularity theorems suggest that it is often advantageous to employ a conformal factor which is proportional to $\theta^{-1}$ \cite{gensing}, thereby factoring out the dominant scale, as well as providing the conformal factor with the dimension ``length'' and rendering the conformal metric dimensionless \cite{RU05,gensing}. This choice yields a state space and a set of regularised dimensionless equations for expansion-normalised variables, which are equivalent to the EFE and from which the dynamical equations for the conformal factor, i.e., equivalently $\theta$, decouple. It is therefore the reduced dimensionless system which contains the essential physical dynamics (this is related to the infamous conformal mode problem which appears in quantum gravity approaches). This choice, furthermore, underlines the role of \emph{self-similar}, i.e., scale-invariant, solutions as important building blocks for the understanding of non-scale-invariant solutions since the former correspond to equilibrium points in state space (hence the dynamical systems approach is also referred to as the \emph{conformal scale-invariant method}) \cite{DS,CC}. Whereas the conformal factor of the dynamical systems approach vanishes towards a singularity and diverges in the ever-expanding case, the conformal factor of the AFEU/AFS-framework carries the diverging properties of the metric in both scenarios. The central object of the dynamical systems approach is the conformal contravariant 3-metric (or equivalently the Hubble-normalised spatial frame components) and the analysis focuses on the regularised EFE and an associated state space. The central objects of the new framework presented in this article, on the other hand, are the conformal covariant metric, the conformal factor and, when it comes to matter dynamics, the conformal reference fluid, and the focus of the new approach lies rather in the conformal geometry. 

The asymptotics further emphasise the ``duality''. The Hubble-normalisation in dynamical systems goes in hand with introducing a new dimensionless unphysical time parameter which takes an infinite value at asymptotic states \cite{DS,UEWE}. This stands in contrast to our framework, which maps the cosmological futures to finite cosmic time values (i.e., $\bar{T}=0$). The dynamical systems approach therefore effectively pushes the singularity away to infinity, while we adopt the strategy from asymptotically-flat space-times of ``pulling in'' the cosmological futures to a finite place. Finally, the asymptotic behaviour of the conformal spatial metric highlights the ``duality'' of the two approaches. On the one hand, the conformal covariant spatial metric of the dynamical systems method becomes singular on approach to a generic initial singularity (and to an IPS as well) while the conformal contravariant spatial metric becomes degenerate \cite{UEWE,LEUW,gensing}. This reflects asymptotic causal properties and is related to \emph{asymptotic silence} (see below). On the other hand, in a suitable frame, the conditions of both the AFEU- and AFS-definition translate into the opposite behaviour, namely a degenerate conformal covariant spatial metric and a singular conformal contravariant spatial metric. We therefore see, that degeneracy already plays an essential role in another conformal approach, even though this degeneracy is dual to the one introduced in the present article.

The dynamical systems approach has advantages when performing an asymptotic analysis of the EFE in the state space, since, due to the regularisation, one is dealing with finite-valued variables. This has been successfully employed in the spatially homogeneous context \cite{DS} and was promisingly extended to the inhomogeneous case and the generic analysis of cosmological initial singularities \cite{UEWE,LEUW,gensing,AELU}. Nevertheless, this technique also has disadvantages when compared to the newly constructed framework of this paper. For instance, the conformal structure of the IPS leads to advantages when formulating the initial value problem (see \cite{GW1985,GoodeColeyWainwright1992,Anguige} and references therein) and the non-singular behaviour of the conformal covariant metric of the AFEU/AFS-framework might offer advantages for numerical simulations when compared to the dynamical systems approach. Lastly, the new conformal structures, certainly, provide advantages for the geometric discussion of asymptotic structures, as is explicitly illustrated by the possibility of deriving results concerning the notion of strong curvature in the AFS scenario \cite{HoehnScott2}. 

The ``dual'' relationship between these two approaches is thus of considerable interest to relativistic cosmology and, consequently, warrants more research.

The dynamical systems techniques are, moreover, mostly applied to the simplified case of a perfect fluid with a linear barotropic equation of state (polytrope) in which case the inhomogeneous cosmologies admitting an IPS could be accommodated in the state space \cite{LEUW}. Analogously, it should also be possible to accommodate and analyse cosmologies satisfying the presented new definitions in the dynamical systems state space, once restricted to polytropic perfect fluids and a suitable reformulation of the new definitions has been performed. 

We now briefly return to the notion of \emph{self-similarity} and \emph{asymptotic silence}. The homogeneous example cosmologies presented in this article are known to be asymptotically self-similar which means that these solutions are asymptotic to a self-similar, i.e., scale-invariant, solution of the EFE \cite{DS}. Since it has recently been shown that \emph{asymptotic self-similarity breaking} occurs generically (in different ways) at both the initial singularity and towards the cosmological future for spatially homogeneous cosmologies \cite{Wain99,HHTW,WHU,CC}, it seems reasonable that this state of affairs will hold in the more general inhomogeneous context as well, and therefore one might wonder whether the definitions of an AFEU and an AFS, together with their example cosmologies, may be misleading for the generic case. In our forthcoming article \cite{HoehnScott2}, however, we will display an inhomogeneous AFEU-example model, which is \emph{not} asymptotically self-similar. The conformal framework presented in this article is hence indeed relevant to a rather generic scenario. The general significance of asymptotic self-similarity breaking in AFEU- and AFS-cosmologies is an interesting topic and requires further investigation.

Likewise, it is of interest to elaborate the relation of the AFS to \emph{asymptotic silence} and its breaking. Asymptotic silence means that communication is interrupted due to vanishing particle horizons which form on approach to the singularity as a consequence of strong gravity effects and is related to the vanishing of the Hubble-normalised spatial frame components \cite{UEWE,gensing} and conformal properties \cite{RU05}. The Kantowski-Sachs solution, i.e., the AFS example, is asymptotically silent at $\bar{T}=0$, since the Hubble-normalised spatial frame components vanish. The opposite scenario, namely asymptotic silence breaking, has been found at special initial singularities \cite{LUW2006}. Current research indicates, however, that generic initial singularities are asymptotically silent \cite{UEWE,gensing,AELU}, while the possibility of asymptotic silence breaking at a future singularity has not been sufficiently investigated to already discuss the relevance of this scenario to the AFS at this point. We leave this question open for future research. 

Finally, since, from our current knowledge
of the universe, it is not clear whether our universe will expand
indefinitely or recollapse in finite future time, we suggest that
the definitions of an AFEU and an AFS should both be regarded as a
completion to the definition of an IPS, and conjecture that the
conjunction of these definitions provides a possible version of a
complete mathematical formalisation of \emph{quiescent cosmology}
and the WCH.

\appendix

\section{Technical properties relating to conformal factors}\label{conffac}

A number of technical properties of a monotonic conformal factor,
which is a function of a cosmic time function $T$, are necessary
for the derivation and proofs of several implications of the
employed conformal structures, e.g., of the results of section
\ref{confiso}, which will be proven in the following appendix. The
quick reader may skip this rather technical appendix and jump
straight to \ref{proofs} and refer back to the following lemmas
once going through the proofs of the above mentioned results.

In this appendix we will examine the existence and the possible values of the limit of several ratios of the conformal factor and its derivatives. Some technicalities, concerning the conformal factor in the new definitions of section \ref{newdefs}, will subsequently become clear.

We will cover the cases of cosmic time functions which approach $0$ from above or below and conformal factors which diverge or vanish at $T=0$. Recall that if the conformal factor and the cosmic time function are denoted with a $\ \bar{} \ $, we are specifically referring to the situation in which the cosmic time function approaches $0$ from below. We begin by analysing the possible values of $\bar{\lambda}=\mathop {\lim }\limits_{\bar{T} \to 0^{-}}\bar{L}$, where $\bar{L}\equiv\frac{\bar{\Omega}''\bar{\Omega}}{(\bar{\Omega}')^{2}}$, for the case $\bar{\Omega}(0)=\infty$.
\begin{lem}\label{conffac1.lem} {If $\bar{T}$ is a cosmic time function which approaches $0$ from below, and
\begin{enumerate}
\item $\bar{\Omega}\left(\bar{T}\right)$ is strictly monotonically increasing, positive and $C^{2}$ on $I=[-c,0)$ (where $c>0$),
\item $\bar{\Omega}\left(0\right)=\infty$, and
\item $\bar{\lambda}$ as defined above exists,
\end{enumerate}
then $\bar{\lambda}\geq 1$.}
\end{lem}
\begin{proof} Define $\phi =\ln\bar{\Omega}$ on $I$. Then
\begin{eqnarray}
\phi '=\frac{\bar{\Omega} '}{\bar{\Omega}} \qquad \mbox{and}
\qquad \phi ''=\bar{\Omega}^{-2}\left[\bar{\Omega} ''\bar{\Omega}
-(\bar{\Omega }')^{2}\right].\label{lam1.eqn}
\end{eqnarray}

From condition (i) it follows immediately that $\phi$ is a
strictly monotonically increasing function on $I$.

Suppose $\bar{\lambda} <1$. Since $\bar{\Omega}\left(\bar{T}\right)$ is $C^{2}$ and $\bar{\Omega}'\left(\bar{T}\right)>0$ on $I$, it follows that $\bar{L}$ is continuous on $I$. So $\exists$ $-\eta$ $\in$ $I$, such that $\bar{L}\left(\bar{T}\right)<1$ $\forall$ $\bar{T}\in J=[-\eta,0)$.
\begin{eqnarray}\fl
\Rightarrow \qquad \frac{ \bar{\Omega}\bar{\Omega}
''}{(\bar{\Omega} ')^{2}}<1 \qquad \Leftrightarrow \qquad
\bar{\Omega} ''\bar{\Omega} - (\bar{\Omega }')^{2}<0 \qquad
\Leftrightarrow \qquad \phi ''<0,
\end{eqnarray}
i.e.,\ $\phi '$ is a strictly monotonically decreasing function on $J$ and thus $\phi '$ is bounded above by $\phi '\left(-\eta\right)$ on $J$. With the help of the mean value theorem we find that
\begin{eqnarray}
\frac{\phi(\bar{T})-\phi(-\eta)}{\bar{T}+\eta}&=&\phi '\left(\xi\right)\leq\phi '\left(-\eta\right) \qquad \mbox{where} \qquad \xi\in[-\eta,\bar{T}]\\
\Rightarrow \qquad \phi\left(\bar{T}\right)
&<&\phi\left(-\eta\right)+\eta\phi '\left(-\eta\right) \qquad
\forall \qquad \bar{T}\in J.\label{lam2.eqn}
\end{eqnarray}
The two terms on the r.h.s.\ of (\ref{lam2.eqn}) are both finite
by construction, therefore $\phi$ is bounded above on $J$. This is
a contradiction to condition (ii) which says that
$\phi\rightarrow\infty$ as $\bar{T}\rightarrow 0^{-}$.

Hence, $\bar{\lambda}\geq 1$.
\end{proof}
A similar result holds for $\bar{\Omega}(0)=0$.
\begin{lem} \label{conffac2.lem}{If $\bar{T}$ is a cosmic time function which approaches $0$ from below, and
\begin{enumerate}
\item $\bar{\Omega}\left(\bar{T}\right)$ is strictly monotonically decreasing, positive and at least $C^{2}$ on $I=[-c,0)$ (where $c>0$), and $C^{0}$ on $[-c,0]$,
\item $\bar{\Omega}\left(0\right)=0$, and
\item $\bar{\lambda}$ as defined above exists,
\end{enumerate}
then $\bar{\lambda}\leq 1$.}
\end{lem}
\begin{proof} As before, define $\phi =\ln\bar{\Omega}$ on $I$. Condition (i) immediately implies that $\phi$ is a strictly monotonically decreasing function on $I$.

Suppose $\bar{\lambda} >1$. Since $\bar{\Omega}\left(\bar{T}\right)$ is $C^{2}$ and $\bar{\Omega}'\left(\bar{T}\right)<0$ on $I$, it follows that $\bar{L}$ is continuous on $I$. So $\exists$ $-\eta$ $\in$ $I$, such that $\bar{L}\left(\bar{T}\right)>1$ $\forall$ $\bar{T}\in J=[-\eta,0)$. By the same arguments as in the previous lemma, this implies $\phi ''>0$, i.e., $\phi '$ is a strictly monotonically increasing function on $J$ and thus $\phi '$ is bounded below by $\phi '\left(-\eta\right)$ on $J$. The mean value theorem implies
\begin{eqnarray}
\frac{\phi(\bar{T})-\phi(-\eta)}{\bar{T}+\eta}&=&\phi '\left(\xi\right)\geq\phi '\left(-\eta\right) \qquad \mbox{where} \qquad \xi\in[-\eta,\bar{T}]\\
\Rightarrow \qquad \phi\left(\bar{T}\right)
&>&\phi\left(-\eta\right)+\eta\phi '\left(-\eta\right) \qquad
\forall \qquad \bar{T}\in J.\label{lam4.eqn}
\end{eqnarray}
The two terms on the r.h.s.\ of (\ref{lam4.eqn}) are both finite
by construction, therefore $\phi$ is bounded below on $J$. This
contradicts condition (ii) which indicates that $\phi\rightarrow
-\infty$ as $\bar{T}\rightarrow 0^{-}$.

Hence, $\bar{\lambda}\leq 1$.
\end{proof}

\begin{rem}\label{conffac1.rem}{Scott \emph{\cite{GVR}} has already proven that if $T\rightarrow 0^{+}$, $\Omega(0)=0$ and $\Omega$ is positive and at least $C^{2}$ on $(0,c]$, $c>0$, and continuous on $[0,c]$, then $\lambda\leq 1$. If one is aware of the signs, by Lemma \ref{conffac1.lem} one can easily verify that under the same conditions, and $\Omega(0)=\infty$, one again finds $\lambda\geq 1$.}
\end{rem}

The limit of $\frac{\bar{\Omega}'}{\bar{\Omega}}$ as $\bar{T}\rightarrow 0^{-}$ turns out to be an essential property in the derivation of some results. We will first prove the existence and the value of this limit for the case $\bar{\Omega}(0)=\infty$.

\begin{lem}\label{conffac3.lem}{Let $\bar{T}$ be a cosmic time function which approaches $0$ from below and let $\bar{\Omega}(\bar{T})$ be positive, strictly monotonically increasing and $C^{2}$ on some interval $[-c,0)$, $c>0$. If, furthermore, $\bar{\Omega}(0)=\infty$ and $\bar{\lambda}$ as defined above exists, $\bar{\lambda}\neq 1$, then $\mathop {\lim }\limits_{\bar{T} \to 0^{-}}\frac{\bar{\Omega}'}{\bar{\Omega}}\mbox{\emph{(exists)}}=+\infty$.}
\end{lem}

\begin{proof}
By Lemma \ref{conffac1.lem} we know that, in this case, $\bar{\lambda}>1$, and proceeding as in the proof of Lemma \ref{conffac1.lem} we find an interval $[-\eta,0)$, $\eta>0$, on which $\phi ''>0$, i.e., on which $\phi '$ is a strictly monotonically increasing function. Thus, $\mathop {\lim }\limits_{\bar{T} \to 0^{-}}\phi'=\mathop {\lim }\limits_{\bar{T} \to 0^{-}}\frac{\bar{\Omega}'}{\bar{\Omega}}$ exists. We note that  $\phi '>0$ on $[-\eta,0)$, which implies that $\mathop {\lim }\limits_{\bar{T} \to 0^{-}}\phi'>0$.

Now assume that
\begin{eqnarray}
\mathop {\lim }\limits_{\bar{T} \to 0^{-}}\frac{\bar{\Omega}
'}{\bar{\Omega}}=k, \qquad \mbox{where} \qquad 0<k<\infty.
\end{eqnarray}
It follows that $\displaystyle\int^0_{-\eta}\frac{\bar{\Omega}'}{\bar{\Omega}} \, d\bar{T}$ is bounded. If we consider the function $F\left(\bar{T}\right) = \ln\bar{\Omega}$, $F'\left(\bar{T}\right) = \frac{\bar{\Omega}'}{\bar{\Omega}}$, we encounter a contradiction, however, since $F\left(\bar{T}\right) \rightarrow \infty$ as $\bar{T}\rightarrow 0^{-}$.

Thus $\mathop {\lim }\limits_{\bar{T} \to 0^{-}}\frac{\bar{\Omega}'}{\bar{\Omega}}=+\infty$.
\end{proof}

An almost identical result holds for the case $\bar{\Omega}(0)=0$.

\begin{lem}\label{conffac4.lem}{Let $\bar{T}$ be a cosmic time function which approaches $0$ from below and let $\bar{\Omega}(\bar{T})$ be positive, strictly monotonically decreasing and $C^{2}$ on some interval $[-c,0)$ and continuous on $[-c,0]$, $c>0$. If, furthermore, $\bar{\Omega}(0)=0$ and $\bar{\lambda}$ as defined above exists, $\bar{\lambda}\neq 1$, then $\mathop {\lim }\limits_{\bar{T} \to 0^{-}}\frac{\bar{\Omega}'}{\bar{\Omega}}\mbox{\emph{(exists)}}=-\infty$.}
\end{lem}

\begin{proof}
Lemma \ref{conffac2.lem} implies that $\bar{\lambda}<1$ and proceeding as in the proof of Lemma \ref{conffac2.lem} we find that $\exists$ $-\eta\in[-c,0)$ such that $\phi ''<0$ on $[-\eta,0)$, i.e., such that $\phi '$ is a strictly monotonically decreasing function on $[-\eta,0)$. Hence, $\mathop {\lim }\limits_{\bar{T} \to 0^{-}}\phi '=\mathop {\lim }\limits_{\bar{T} \to 0^{-}}\frac{\bar{\Omega}'}{\bar{\Omega}}$ exists. We note that  $\phi '<0$ on $[-\eta,0)$, which implies that $\mathop {\lim }\limits_{\bar{T} \to 0^{-}}\phi'<0$.

Now suppose that
\begin{eqnarray}
\mathop {\lim }\limits_{\bar{T} \to 0^{-}}\frac{\bar{\Omega}
'}{\bar{\Omega}}=k, \qquad \mbox{where} \qquad -\infty<k<0.
\end{eqnarray}
It follows that $\displaystyle\int^0_{-\eta}\frac{\bar{\Omega}'}{\bar{\Omega}} \, d\bar{T}$ is bounded. If we consider the function $F\left(\bar{T}\right) = \ln\bar{\Omega}$, $F'\left(\bar{T}\right) = \frac{\bar{\Omega}'}{\bar{\Omega}}$, we encounter a contradiction, however, since $F\left(\bar{T}\right) \rightarrow -\infty$ as $\bar{T}\rightarrow 0^{-}$.

Thus $\mathop {\lim }\limits_{\bar{T} \to 0^{-}}\frac{\bar{\Omega}'}{\bar{\Omega}}=-\infty$.
\end{proof}

The existence of these limits for $\bar{\lambda}=1$ is not evident. This case is therefore omitted in the new definitions in section \ref{newdefs} (as was the case in Definition \ref{ISdef1.def}).

\begin{rem}\label{conffac2.rem}
{If $T\rightarrow 0^{+}$, Lemmas \ref{conffac3.lem} and
\ref{conffac4.lem} are clearly also true if the signs in the
assertions are reversed. }
\end{rem}

The cases $\bar{\Omega}(0)=\infty$ and $\bar{\Omega}(0)=0$ lead to rather different behaviours of the ratio $\bar{M}\equiv\frac{\bar{\Omega}'}{\bar{\Omega}^{2}}$, which will be discussed next.

\begin{lem}\label{conffac5.lem}{Let $\bar{T}$ be a cosmic time function which approaches $0$ from below and let $\bar{\Omega}(\bar{T})$, with $\bar{\Omega}(0)=\infty$, be positive, strictly monotonically increasing and $C^{2}$ on some interval $[-c,0)$, $c>0$. If $\bar{\lambda}$ exists and $\bar{\lambda}\neq 1,2$, then $\mathop {\lim }\limits_{\bar{T} \to 0^{-}}\bar{M}\mbox{\emph{(exists)}}=\kappa$, and $0<\kappa\leq\infty$ if $\bar{\lambda}>2$ and $0\leq\kappa<\infty$ if $\bar{\lambda}<2$.}
\end{lem}

\begin{proof}
The conditions imply $\bar{M}>0$ on $[-c,0)$. Furthermore,
\begin{eqnarray}
\bar{M}'=\frac{\bar{\Omega}''\bar{\Omega}-2(\bar{\Omega}')^{2}}{\bar{\Omega}^{3}}=(\bar{L}-2)\frac{(\bar{\Omega}')^{2}}{\bar{\Omega}^{3}}\label{lem6.eqn}
\end{eqnarray}
and $\frac{(\bar{\Omega}')^{2}}{\bar{\Omega}^{3}}>0$ on $[-c,0)$ and $\bar{\lambda}>1$. There are two cases:
\begin{description}
\item[(a)] $\bar{\lambda}>2$, then $\exists$ $d>0$, such that $\bar{M}'>0$ on $[-d,0)$. Now $\bar{M}>0$, $\bar{M}'>0$ $\Rightarrow$ $\mathop {\lim }\limits_{\bar{T} \to 0^{-}}\bar{M}\mbox{(exists)}=\alpha\in\mathbb{R}^{+}\cup\{+\infty\}$.
\item[(b)] $\bar{\lambda}<2$, then $\exists$ $d>0$, such that $\bar{M}'<0$ on $[-d,0)$. Now $\bar{M}>0$, $\bar{M}'<0$ $\Rightarrow$ $\mathop {\lim }\limits_{\bar{T} \to 0^{-}}\bar{M}\mbox{(exists)}=\beta\in\mathbb{R}^{+}\cup\{0\}$.
\end{description}
\end{proof}


\begin{lem} \label{conffac6.lem}{If all conditions of Lemma \ref{conffac5.lem} are satisfied, except that here we assume that $\bar{\lambda}=2$ and $\bar{L}$ is a strictly monotonic function on $[-c,0)$, then $\mathop {\lim }\limits_{\bar{T} \to 0^{-}}\bar{M}\mbox{\emph{(exists)}}\geq0$.}
\end{lem}

\begin{proof}
Since $\bar{L}$ is continuous and strictly monotonic on $[-c,0)$, there are two possibilities:
\begin{enumerate}
\item $\bar{L}\rightarrow 2^{+}$ as $\bar{T}\rightarrow 0^{-}$. By equation (\ref{lem6.eqn}) $\exists$ $d>0$ such that $\bar{M}>0$ and $\bar{M}'>0$ on $[-d,0)$. Thus, $\mathop {\lim }\limits_{\bar{T} \to 0^{-}}\bar{M}\mbox{(exists)}=\alpha\in \mathbb{R}^{+}\cup\{+\infty\}$.
\item $\bar{L}\rightarrow 2^{-}$ as $\bar{T}\rightarrow 0^{-}$. By equation (\ref{lem6.eqn}) $\exists$ $d>0$ such that $\bar{M}>0$ and $\bar{M}'<0$ on $[-d,0)$. Thus, $\mathop {\lim }\limits_{\bar{T} \to 0^{-}}\bar{M}\mbox{(exists)}=\beta\in\mathbb{R}^{+}\cup\{0\}$.
\end{enumerate}

\end{proof}

\begin{rem}\label{conffac4.rem}{In the case $T\rightarrow 0^{+}$ Lemmas \ref{conffac5.lem} and \ref{conffac6.lem} are clearly also true if the signs in the limit of $M$ are reversed.}
\end{rem}

For the case $\bar{\Omega}(0)=0$ the analyses are much simpler.

\begin{lem}\label{conffac7.lem}{Let $\bar{T}$ be a cosmic time function which approaches $0$ from below and let $\bar{\Omega}(\bar{T})$, with $\bar{\Omega}(0)=0$, be positive, strictly monotonically decreasing and $C^{2}$ on some interval $[-c,0)$ and continuous on $[-c,0]$, $c>0$. If $\bar{\lambda}$ exists and $\bar{\lambda}\neq 1$, then $\mathop {\lim }\limits_{\bar{T} \to 0^{-}}\bar{M}\mbox{\emph{(exists)}}=-\infty$.}
\end{lem}

\begin{proof}
By Lemma \ref{conffac4.lem} we know that $\mathop {\lim
}\limits_{\bar{T} \to
0^{-}}\frac{\bar{\Omega}'}{\bar{\Omega}}\mbox{(exists)}=-\infty$.
Since $\bar{\Omega}(0)=0$, clearly $\mathop {\lim
}\limits_{\bar{T} \to 0^{-}}\bar{M}\mbox{(exists)}=-\infty$.
\end{proof}

Finally, we examine the function $\bar{N}\equiv\bar{L}^{2}\bar{M}^{4}$.

\begin{lem}\label{conffac8.lem} {Let $\bar{T}$ be a cosmic time function which approaches $0$ from below and let, furthermore, $\bar{\Omega}(\bar{T})$ be a conformal factor which satisfies $\bar{\Omega}(0)=\infty$ and $\bar{\Omega}$ is positive, strictly monotonically increasing and $C^{2}$ on $[-c,0)$, $c>0$. Then, if $\bar{\lambda}$ exists and $\bar{\lambda}\neq 1,2$, $\mathop {\lim }\limits_{\bar{T} \to 0^{-}}\bar{N}\mbox{\emph{(exists)}}=\varrho$, and $0<\varrho\leq\infty$ if $\bar{\lambda}>2$ and $0\leq\varrho<\infty$ if $\bar{\lambda}<2$.}
\end{lem}

\begin{proof}
We know that $1<\bar{\lambda}\leq\infty$ and $\bar{\lambda}\neq 2$. The same cases apply as in Lemma \ref{conffac5.lem}.
\begin{description}
\item[(a)] $\bar{\lambda}>2$, $4<\bar{\lambda}^{2}\leq\infty$. By Lemma \ref{conffac5.lem} $\mathop {\lim }\limits_{\bar{T} \to 0^{-}}\bar{M}\mbox{(exists)}=\alpha\in\mathbb{R}^{+}\cup\{+\infty\}$. Thus, $\mathop {\lim }\limits_{\bar{T} \to 0^{-}}\bar{N}\mbox{(exists)}=\alpha '\in\mathbb{R}^{+}\cup\{+\infty\}$.
\item[(b)] $\bar{\lambda}<2$, $1<\bar{\lambda}^{2}<4$. From Lemma \ref{conffac5.lem} we know that $\mathop {\lim }\limits_{\bar{T} \to 0^{-}}\bar{M}\mbox{(exists)}=\beta\in\mathbb{R}^{+}\cup\{0\}$. Hence, $\mathop {\lim }\limits_{\bar{T} \to 0^{-}}\bar{N}\mbox{(exists)}=\beta '\in\mathbb{R}^{+}\cup\{0\}$.
\end{description}
\end{proof}


\begin{rem}\label{conffac6.rem}{Remark \ref{conffac4.rem} implies that the result of Lemma \ref{conffac8.lem} is still valid for $N$ if $T\rightarrow 0^{+}$ and $\Omega(0)=\infty$.}
\end{rem}

\begin{lem}\label{conffac9.lem}{Let $\bar{T}$ be a cosmic time function which approaches $0$ from below and let, furthermore, $\bar{\Omega}(\bar{T})$ be a conformal factor which satisfies $\bar{\Omega}(0)=0$ and $\bar{\Omega}$ is positive, strictly monotonically decreasing and $C^{2}$ on $[-c,0)$ and continuous on $[-c,0]$, $c>0$. Then, if $\bar{\lambda}$ exists and $\bar{\lambda}\neq 0,1$, $\mathop {\lim }\limits_{\bar{T} \to 0^{-}}\bar{N}\mbox{\emph{(exists)}}=+\infty$.}
\end{lem}

\begin{proof}
We have $\bar{\lambda}<1$ and $\bar{\lambda}\neq0$, i.e.,
$0<\bar{\lambda}^{2}\leq+\infty$. By Lemma \ref{conffac7.lem},
$\mathop {\lim }\limits_{\bar{T} \to
0^{-}}\bar{M}\mbox{(exists)}=-\infty$. Thus, $\mathop {\lim
}\limits_{\bar{T} \to 0^{-}}\bar{N}\mbox{(exists)}=+\infty$.
\end{proof}

The special cases of $\bar{\lambda}=0,1,2$ in the above lemmas require special treatment and warrant further investigation.

\section{Proofs of the results of section \ref{confiso}}\label{proofs}

Since we are dealing with cosmic time functions which approach $0$
from above or below in the same proofs, we choose the notation of
Definition \ref{ISdef1.def} and section \ref{confiso} for
simplicity in this appendix, i.e., we denote the cosmic time
function and the conformal factor without a $\ \bar{} \ $ and the
quantities of the conformal space-time with a $\ \tilde{} \ $ for
both future and past scenarios.

We begin by proving Theorem \ref{K.thm}.

\begin{proof}
We will analyse $K$ directly. Using the following well-known
conformal relation for the Ricci tensor (e.g., see
\cite{GW1985,HawkEll1973,Geoffthesis}, a colon denotes covariant
differentiation with respect to $\mathbf{\tilde{g}}$)
\begin{align}
 R_{ab}=&\left(\frac{{\Omega}'}{{\Omega}}\right)^{2}\{2\left(2-{L}\right){T}_{,a}{T}^{,c}\tilde{g}_{cb}-\left(1+L\right)\tilde{g}_{ab}T_{,c}T_{,d}\tilde{g}^{cd}\} \\
&-\left(\frac{\Omega '
}{\Omega}\right)\{2T_{:ac}{\delta^{c}}_{b}+\tilde{g}_{ab}T_{:cd}\tilde{g}^{cd}\}+\tilde{R}_{ab},
\end{align}
we obtain
\begin{align}
R_{ab}R^{ab}= & \Omega^{-4}\left[\left(\frac{\Omega ' }{\Omega}\right)^{4}12\left(T_{,a}T^{,a}\right)^{2}\left(L^{2}-L+1\right)\right.\\
&-2\left(\frac{\Omega ' }{\Omega}\right)^{3}\{\left(8-4L\right)T^{,a}T^{,b}T_{:ba}-\left(8L+2\right)T_{,a}T^{,a}{T_{:b}} {^{b}}\}\\
&+\left(\frac{\Omega ' }{\Omega}\right)^{2}\{4T_{:ab}T^{:ba}+8\left({T_{:a}} {^{a}}\right)^{2}+4\left(2-L\right)T_{,a}T_{,b}\tilde{R}^{ba}-2\left(1+L\right)\tilde{R}T_{,a}T^{,a}\}\\
&-\left.2\left(\frac{\Omega '
}{\Omega}\right)\{2T_{:ab}\tilde{R}^{ba}+\tilde{R}{T_{:a}}
{^{a}}\}+\tilde{R}_{ab}\tilde{R}^{ab}\right]. \label{weric1.eqn}
\end{align}

Furthermore, $g_{ab}=\Omega^{2}\tilde{g}_{ab}$, $g^{ab}=\Omega^{-2}\tilde{g}^{ab}$ and ${C^{a}}_{bcd}={{\tilde{C}}^{a}}{_{bcd}}$. Thus,
\begin{eqnarray}
C_{abcd}C^{abcd}=\Omega^{-4}\tilde{C}_{abcd}\tilde{C}^{abcd}.\label{weric2.eqn}
\end{eqnarray}

By the conjunction of equations (\ref{weric1.eqn}) and (\ref{weric2.eqn}) we find
\begin{eqnarray}\fl
K=\frac{\tilde{C}_{abcd}\tilde{C}^{abcd}}{\left(\frac{\Omega ' }{\Omega}\right)^{4}12\left(T_{,e}T^{,e}\right)^{2}\left(L^{2}-L+1\right)-2\left(\frac{\Omega ' }{\Omega}\right)^{3}\left[\cdots\right]+\cdots+\tilde{R}_{ef}\tilde{R}^{ef}}  \label{weric3.eqn}\,.
\end{eqnarray}

$\tilde{R}_{ef}\tilde{R}^{ef}$, as well as $\tilde{C}_{abcd}\tilde{C}^{abcd}$, and the derivatives of $T$ are well-behaved at $T=0$, because we have required that the conformal metric be non-degenerate and at least $C^{2}$ on an open neighbourhood of $T=0$. Additionally, since a cosmic time function has a nowhere vanishing timelike derivative, we have $T_{,a}T^{,a}\neq 0$ on an open neighbourhood of $T=0$.

In Lemmas \ref{conffac3.lem} and \ref{conffac4.lem} and Remark
\ref{conffac2.rem} we have seen that under the imposed conditions,
$\mathop {\lim }\limits_{T \to 0^{\pm}}\frac{\Omega
'}{\Omega}\mbox{(exists)}=\pm\infty$ (the sign depending on
whether $T \to 0^{\pm}$ and the value of $\Omega(0)$).
Asymptotically, the sign of $\frac{\Omega '}{\Omega}$ will not
matter, since the dominant term in (\ref{weric3.eqn}) is an even
power of $\frac{\Omega '}{\Omega}$. Then, as $T\rightarrow
0^{\pm}$, we have two possible cases for ${\lambda}=\mathop {\lim
}\limits_{T \to 0^{\pm}}L$:
\begin{enumerate}
\item ${\lambda}$ finite or zero. Clearly,
\begin{eqnarray}
K\approx\frac{\tilde{C}_{abcd}\tilde{C}^{abcd}}{\left(\frac{\Omega ' }{\Omega}\right)^{4}12\left(T_{,e}T^{,e}\right)^{2}\left({\lambda}^{2}-{\lambda}+1\right)}.
\end{eqnarray}
Since ${\lambda}^{2}-{\lambda}+1>0$ $\forall$ ${\lambda}\in\mathbb{R}$, we immediately have that $\mathop {\lim }\limits_{T \to 0^{\pm} }K=0$.
\item  ${\lambda}=\pm\infty$. Then
\begin{eqnarray}
K\approx\frac{\tilde{C}_{abcd}\tilde{C}^{abcd}}{\left(\frac{\Omega ' }{\Omega}\right)^{4}12\left(T_{,e}T^{,e}\right)^{2}L^{2}},
\end{eqnarray}
and consequently $\mathop {\lim }\limits_{T \to 0^{\pm} }K=0$.
\end{enumerate}
\end{proof}

This proof is the source of the following results. We next prove Corollary \ref{regmet1.thm}.

\begin{proof}
The proof is evident by the regularity in
$(\tilde{\mathcal{M}},\mathbf{\tilde{g}})$, the expressions given
in (\ref{weric1.eqn}) and (\ref{weric2.eqn}), and the fact that
$\mathop {\lim }\limits_{{T} \to 0^{\pm} }{\Omega}({T})=0$ as well
as $\mathop {\lim }\limits_{T \to 0^{\pm}}\frac{\Omega
'}{\Omega}\mbox{(exists)}=\pm\infty$.
\end{proof}

Proof of Corollary \ref{regmet3.thm}.

\begin{proof}
By (\ref{weric1.eqn}) in the proof of Theorem \ref{K.thm}, we find
the following asymptotic relationship under the required
conditions:
\begin{eqnarray}
R_{ab}R^{ab}&\approx&{\Omega}^{-4}\left(\frac{{\Omega} '
}{{\Omega}}\right)^{4}12\left({T}_{,a}{T}^{,a}\right)^{2}\left({\lambda}^{2}-{\lambda}+1\right)\cr
&=&{M}^{4}12\left({T}_{,a}{T}^{,a}\right)^{2}\left({\lambda}^{2}-{\lambda}+1\right).
\end{eqnarray}
Since ${\lambda}^{2}-{\lambda}+1>0$ $\forall$ ${\lambda}\in\mathbb{R}$ and ${T}_{,a}{T}^{,a}\neq 0$, the relationship only depends on the possible value of ${M}$. In Lemmas \ref{conffac5.lem} and \ref{conffac6.lem} and Remark \ref{conffac4.rem} we have proven the existence and the possible values of ${M}_{0}$, in accordance with the assertion.
\end{proof}

Proof of Corollary \ref{regmet4.thm}.

\begin{proof}
Under these conditions, (\ref{weric1.eqn}) implies the following
asymptotic relationship:
\begin{eqnarray}
R_{ab}R^{ab}&\approx&{\Omega}^{-4}\left(\frac{{\Omega} '
}{{\Omega}}\right)^{4}{L}^{2}12\left({T}_{,a}{T}^{,a}\right)^{2}\cr
&=&{N}12\left({T}_{,a}{T}^{,a}\right)^{2},
\end{eqnarray}
which merely depends on ${N}$ since ${T}_{,a}{T}^{,a}\neq 0$. Lemma \ref{conffac8.lem} and Remark \ref{conffac6.rem} prove $\mathop {\lim }\limits_{{T} \to 0^{\pm} }{N}=\infty$.
\end{proof}

Proof of Corollary \ref{regmet5.thm}.

\begin{proof}
(\ref{weric2.eqn}) implies the assertion, since
$\tilde{C}_{abcd}\tilde{C}^{abcd}$ is well-behaved on an open
neighbourhood of ${T}=0$ and ${\Omega}(0)=\infty$.
\end{proof}

We now present the proof of the asymptotic kinematic isotropy (Theorem \ref{asykiniso.thm}), which is a generalisation of Theorem 3.3 of \cite{GW1985}.

\begin{proof}
We employ the conformal relations for the kinematic quantities
(e.g., see \cite{GW1985,Geoffthesis}):
\begin{eqnarray}
\theta&=&\Omega^{-1}\left[3\frac{\Omega'}{\Omega}T_{,a}\tilde{u}^{a}+\tilde{\theta}\right],\label{regkin1.eqn}\\
\sigma^{2}&=&\Omega^{-2}\tilde{\sigma}^{2}, \qquad  \omega^{2}=\Omega^{-2}\tilde{\omega}^{2},\label{regkin2.eqn}\\
\dot{u}^{a}&=&\Omega^{-2}\left[\dot{\tilde{u}}^{a}+\tilde{h}^{ab}\frac{\Omega'}{\Omega}{T_{,b}}\right],\label{regkin3.eqn}
\end{eqnarray}
where $\tilde{h}_{ab}=\tilde{g}_{ab}+\tilde{u}_a \tilde{u}_b$ denotes the projection operator into the rest space orthogonal to $\tilde{u}^{a}$. The expression $\dot{u}^{a}\dot{u}_{a}$ is given by
\begin{eqnarray}
\dot{u}^{a}\dot{u}_{a}=\Omega^{-2}\left[\left(\frac{\Omega'}{\Omega}\right)^{2}\tilde{h}^{ab}T_{,a}T_{,b}+2\frac{\Omega'}{\Omega}\tilde{h}^{ab}\dot{\tilde{u}}_{a}T_{,b}+\dot{\tilde{u}}^{a}\dot{\tilde{u}}_{a}\right].\label{regkin4.eqn}
\end{eqnarray}
Due to the regularity in the unphysical space-time and the fact
that $\mathop {\lim }\limits_{T \to 0^{\pm}}\frac{\Omega
'}{\Omega}\mbox{(exists)}=\pm\infty$, we obtain the following
expressions, as $T\rightarrow 0^{\pm}$,
\begin{eqnarray}\fl
\frac{\sigma^{2}}{\theta^{2}}=\frac{\tilde{\sigma}^{2}}{\left[3\frac{\Omega'}{\Omega}T_{,a}\tilde{u}^{a}+\tilde{\theta}\right]^{2}}\,,
\qquad
\frac{\omega^{2}}{\theta^{2}}=\frac{\tilde{\omega}^{2}}{\left[3\frac{\Omega'}{\Omega}T_{,a}\tilde{u}^{a}+\tilde{\theta}\right]^{2}}
\,,\qquad
\frac{\dot{u}^{a}\dot{u}_{a}}{\theta^{2}}\approx\frac{\tilde{h}^{bc}T_{,b}T_{,c}}{9(T_{,a}\tilde{u}^{a})^{2}}\,.
\end{eqnarray}

Since $\mathbf{u}$ (and therefore $\mathbf{\tilde{u}}$) is
orthogonal to the slice $T=0$, $\tilde{h}^{bc}T_{,b}T_{,c} = 0$ at
$T=0$. This, together with the regularity in the unphysical
space-time, the fact that $\mathop {\lim }\limits_{T \to
0^{\pm}}\frac{\Omega '}{\Omega}\mbox{(exists)}=\pm\infty$ and that
$T_{,a}\tilde{u}^{a}$ is everywhere non-zero, yields the
assertions of the theorem.
\end{proof}

This proof is the source of the following results. We continue with the proof of Corollary \ref{regkin1.thm}.

\begin{proof}
The expressions given in (\ref{regkin1.eqn}), (\ref{regkin2.eqn})
and (\ref{regkin4.eqn}), the regularity in
$(\tilde{\mathcal{M}},\mathbf{\tilde{g}})$, and the existence and
the value of the limit of $\frac{{\Omega}'}{{\Omega}}$ immediately
imply the assertion.
\end{proof}

Proof of Corollary \ref{regkin3.thm}.

\begin{proof}
By (\ref{regkin2.eqn}), the regularity in
$(\tilde{\mathcal{M}},\mathbf{\tilde{g}})$, and the fact that
$\mathop {\lim }\limits_{T \to 0^{\pm}}{\Omega}(T)=\infty$, we
immediately find the assertion.
\end{proof}

Proof of Corollary \ref{regkin4.thm}

\begin{proof}
(\ref{regkin1.eqn}) and (\ref{regkin4.eqn}) imply the following
asymptotic behaviours as $T\rightarrow 0^{\pm}$:
\begin{eqnarray}
\theta\approx 3{M}{T}_{,a}\tilde{{u}}^{a}, \qquad
\dot{u}^{a}\dot{u}_{a}\approx M^{2}\tilde{h}^{ab}T_{,a}T_{,b},
\end{eqnarray}
which solely depend on ${M}$ since ${T}_{,a}\tilde{{u}}^{a}> 0$,
$\tilde{h}^{ab}T_{,a}T_{,b}>0$ and both are regular on an open
neighbourhood of ${T}=0$. Lemma \ref{conffac5.lem} and Remark
\ref{conffac4.rem} establish the existence and the range of
possible values of ${M}_{0}$, in agreement with the assertion.
\end{proof}

\vspace{1cm}

\begin{center}{\bf  ACKNOWLEDGMENTS}\end{center}
~\\
 P.\ A.\ H.\ acknowledges the support of the German Academic
Exchange Service (DAAD) while undertaking this research at The
Australian National University. The authors thank Michael Ashley
and Benjamin Whale for useful technical discussions relevant to
this work and an anonymous referee for constructive comments.


\begin{thebibliography}{99}
\bibitem{GVR} Scott S M The general vorticity result {\itshape to be
published}
\end{thebibliography}
\end{document}